\documentclass[xcolor=svgnames,11pt,reqno]{amsart}
\usepackage[T1]{fontenc}
\usepackage[round]{natbib}
\usepackage{amsmath}
\usepackage{amssymb}
\usepackage{amsthm}
\usepackage{fancyhdr}
\usepackage{enumitem}
\usepackage{comment}
\usepackage{relsize}
\usepackage{nicefrac}
\usepackage{bbm}
\usepackage{mathrsfs}
\usepackage{graphicx}
\usepackage[utf8]{inputenc}
\usepackage{cancel}
\usepackage{mathtools}
\usepackage{tikz,xcolor}
\usepackage{dirtytalk}
\usetikzlibrary{cd,decorations.pathreplacing,matrix,arrows,positioning,automata,shapes,shadows,calc,fadings,decorations,through,intersections}
\usepackage{float}
\usepackage{subcaption}
\usepackage{multirow}
\usepackage[left=2.5cm,top=2.8cm,right=2.5cm,bottom=2.5cm,head=3cm,headsep=1cm, foot=3cm]{geometry}
\usepackage{hyperref}
\usepackage[skip=2mm,indent=12pt]{parskip}
\definecolor{darkblue}{rgb}{0.1,0.1,0.9}
\definecolor{darkred}{rgb}{0.9,0.1,0.1}

\hypersetup{colorlinks, allcolors= darkblue}

\newtheorem{theorem}{Theorem}[section]
\newtheorem{definition}[theorem]{Definition}
\newtheorem{proposition}[theorem]{Proposition}
\newtheorem{assumption}[theorem]{Assumption}
\newtheorem{corollary}[theorem]{Corollary}
\newtheorem{lemma}[theorem]{Lemma}
\newtheorem{remark}[theorem]{Remark}

\newcommand{\p}{\mathbb{P}}

\newcommand{\cF}{\mathcal{F}}

\newcommand{\cI}{\mathcal{I}}
\newcommand{\cB}{\mathcal{B}}
\newcommand{\cA}{\mathcal{A}}

\newcommand{\cQ}{\mathcal{Q}}

\newcommand{\bbR}{\mathbb{R}}

\renewcommand{\tilde}{\widetilde}

\newcommand{\E}{\mathbb{E}}

\newcommand{\One}{\mathbbm{1}}

\renewcommand{\tilde}{\widetilde}

\DeclareMathOperator*{\argmax}{\arg\max}

\def\ssd{\succcurlyeq_{\hspace{- 0.4 mm} _{SSD}}}

\usepackage{booktabs}
\usepackage{siunitx}
\usepackage[colorinlistoftodos,textsize=small]{todonotes}
\usepackage{cleveref}
\usepackage{caption}
\usepackage{anyfontsize}
\usepackage{yhmath}
\usepackage{comment}

\allowdisplaybreaks
\makeatletter

\newcommand{\Rmnum}[1]{\expandafter\@slowromancap\romannumeral #1@}

\title[Stackelberg Equilibria in Monopoly Insurance Markets with Probability Weighting]{Stackelberg Equilibria in Monopoly Insurance Markets\vspace{0.2cm}\\with Probability Weighting\vspace{0.2cm}}

\thanks{\vspace{0.2cm}}

\thanks{ \vspace{0.2cm}}

\thanks{We are grateful to two anonymous reviewers for comments. Mario Ghossoub acknowledges financial support from the Natural Sciences and Engineering Research Council of Canada (NSERC Grant No.\ 2024-03744). Bin Li acknowledges financial support from (NSERC Grant No.\ 2020-04338). Benxuan Shi acknowledges financial support from the Society of Actuaries through the Hickman Scholars Program.}

\begin{document}

\author[Maria Andraos, Mario Ghossoub, Bin Li, and Benxuan Shi]
{Maria Andraos\vspace{0.1cm}\\ University of Waterloo\vspace{0.6cm}\\
Mario Ghossoub\vspace{0.1cm}\\ University of Waterloo\vspace{0.6cm}\\
Bin Li\vspace{0.1cm}\\ University of Waterloo\vspace{0.6cm}\\
Benxuan Shi\vspace{0.1cm}\\ University of Waterloo
\vspace{0.8cm}\\ \today\vspace{0.2cm}}

\address{{\bf Maria Andraos}: University of Waterloo -- Department of Statistics and Actuarial Science -- 200 University Ave.\ W.\ -- Waterloo, ON, N2L 3G1 -- Canada\vspace{0.1cm}}\email{\href{mailto:mandraos@uwaterloo.ca}{mandraos@uwaterloo.ca}\vspace{0.2cm}}

\address{{\bf Mario Ghossoub}: University of Waterloo -- Department of Statistics and Actuarial Science -- 200 University Ave.\ W.\ -- Waterloo, ON, N2L 3G1 -- Canada\vspace{0.1cm}}\email{\href{mailto:mario.ghossoub@uwaterloo.ca}{mario.ghossoub@uwaterloo.ca}\vspace{0.2cm}}

\address{{\bf Bin Li}: University of Waterloo -- Department of Statistics and Actuarial Science -- 200 University Ave.\ W.\ -- Waterloo, ON, N2L 3G1 -- Canada\vspace{0.1cm}}\email{\href{mailto:bin.li@uwaterloo.ca}{bin.li@uwaterloo.ca}\vspace{0.2cm}}

\address{{\bf Benxuan Shi}: University of Waterloo -- Department of Statistics and Actuarial Science -- 200 University Ave.\ W.\ -- Waterloo, ON, N2L 3G1 -- Canada\vspace{0.1cm}}\email{\href{mailto:benxuan.shi1@uwaterloo.ca}{benxuan.shi1@uwaterloo.ca}}

%============================================
%============================================
%============================================

\begin{abstract}
We study Stackelberg Equilibria (Bowley optima) in a monopolistic centralized sequential-move insurance market, with a profit-maximizing insurer who sets premia using a distortion premium principle, and a single policyholder who seeks to minimize a distortion risk measure. We show that equilibria are characterized as follows: In equilibrium, the optimal indemnity function exhibits a layer-type structure, providing full insurance over any loss layer on which the policyholder is more pessimistic than  the insurer's pricing functional about tail losses; and no insurance coverage over loss layers on which  the policyholder is less pessimistic than the insurer's pricing functional about tail losses. In equilibrium, the optimal pricing distortion function is determined by the policyholder's degree of risk aversion, whereby prices never exceed the policyholder's marginal willingness to insure tail losses. Moreover, we show that both the insurance coverage and the insurer's expected profit increase with the policyholder's degree of risk aversion. Additionally, and echoing recent work in the literature, we show that equilibrium contracts are Pareto efficient, but they do not induce a welfare gain to the policyholder. Conversely, any Pareto-optimal contract that leaves no welfare gain to the policyholder can be obtained as an equilibrium contract. Finally, we consider a few examples of interest that recover some existing results in the literature as special cases of our analysis.  
\end{abstract}

\maketitle

%============================================
%============================================
%============================================

\newpage

\section{Introduction}
\label{SecIntro}
In monopoly insurance markets under perfect information, the classical literature has been mostly interested in characterizing the insurer's profit-maximizing (or welfare-maximizing) insurance contracts, without any consideration for strategic interaction between the insurer and insured. Whenever such strategic considerations are introduced into the market model, the natural framework is a sequential-move game with the insurer as the leader and the insured as the follower, borrowing from the literature on bilateral monopoly. This is a two-stage game, whereby in the first stage, the insured selects their optimal indemnification given a certain pricing mechanism selected by the insurer; and in the second stage, the insurer observes the insured's demand function and sets prices so as to maximize profit or welfare. The associated strategic equilibrium concept has been termed the \textit{Stackelberg equilibrium}, or \textit{Bowley optimum}.

\medskip

Bowley optima were first introduced by \cite{bowley1928} in the context of a bilateral monopoly, and then first applied to insurance markets by \cite{chan1985reinsurer}, in the context of Expected-Utility (EU) preferences with exponential utility functions. Several extensions and/or modifications to this model have subsequently been proposed. For instance, \cite{taylor1992}  extends these findings to more general risk exchanges with EU preferences. \cite{cheung2019risk} consider the case of markets in which the insurer is a risk-neutral Expected utility maximizer, who sets premia using a distortion premium principle, and a insured who seeks to minimize a distortion risk measure (DRM). They assume that the insured's distortion function is either a concave function (indicating risk aversion), or a distortion function corresponding to the Value-at-Risk (VaR) risk measure. \cite{li2021bowley} characterize Stackelberg equilibria when agents use a mean-variance functional. \cite{boonen2021a} and \cite{boonen2021b} examine the effect of information asymmetry in the context of DRMs. \cite{GLS2025} propose an extension in a different direction, by examining the optimal (nonlinear) pricing mechanisms in the market for the class of deductible and coinsurance indemnity functions. \cite{BOONEN2023382} examine the relationship between Bowley optimality and Pareto optimality, under fairly general preferences. \cite{ZHU202324} and \cite{GhossoubZhu2024} provide the first extensions beyond the case of a two-agent insurance market. The former consider a market with multiple insurer having the first move advantage, and one insured; whereas the latter consider the case of one monopoly insurance facing demand from several policyholders. \cite{ANDRAOS2026103210} recently provided a unification and an extension thereof to more general preferences.  

\medskip

Beyond the static framework discussed above, a growing literature has studied Stackelberg equilibria in dynamic settings using stochastic differential game techniques, in which the reinsurer, acting as the leader, and the insurer, acting as the follower, interact sequentially over time. 
In particular, \cite{Chen_Shen_2018} study a continuous-time Stackelberg reinsurance game in which both parties aim to maximize the expected utility of their terminal wealth.
\cite{CHEN2019120} extend this framework to a mean-variance setting. More recent contributions extend this framework to settings involving random horizons, ambiguity aversion, and L\'evy risk processes. In particular, \cite{LI202242} consider the problem with a random time horizon under the mean-variance principle, while \cite{CAO2022128} study a dynamic Stackelberg game between an ambiguity-averse insurance buyer and seller under the spectrally negative L\'evy processes. 
\cite{cao2023stackelberg} extend the model of \cite{CAO2022128} to a general divergence penalty for ambiguity. 
Several extensions also incorporate investment decisions and interactions in stochastic financial markets (e.g., \cite{gu2020optimal}, \cite{wang2020robust}, \cite{bai2021stackelberg}, \cite{bai2022hybrid}, \cite{yang2022robust}), while more recent studies consider dynamic settings with multiple reinsurers (e.g., \cite{CAO2023928}, \cite{cao2024stackelberg}, \cite{chen2020continuous}).

\medskip

In this paper, we consider a monopoly insurance market with a single policyholder. We assume that the policyholder evaluates insurance contracts using a distortion risk measure, following the approach of \cite{ASSA201570}, \cite{cheung2019risk}, or \cite{ZHU202324}, for instance. The insurer is a risk-neutral profit maximizer, who sets premia using a distortion premium principle. To avoid \textit{ex post} moral hazard that might arise from the policyholder's misreporting of the true value of the loss, we impose the customary \textit{no sabotage} condition of \cite{carlierdana2003,CarlierDana2005a} on the set of acceptable indemnity functions. Our framework is most closely related to that of \cite{cheung2019risk}. However, in contrast to their study, which assumes that the policyholder is either strictly risk averse (strictly concave distortion function) or a VaR minimizer, we impose no specific restriction on the curvature of the policyholder's distortion function. This is a significant extension, as it accommodates for several forms of distortion functions that are empirically more relevant than the concave distortion functions, such as the inverse-S-shaped distortion functions of \cite{Tversky1992}, the S-shaped distortion functions of \cite{Prelec98}, or the very flexible class of distortion functions recently introduced by \cite{bleichrodt2023testing}. 

\medskip

To characterize Stackelberg equilibria in our model, we proceed in two steps. First for a fixed distortion function used to determine the distortion premium principle, we determine the optimal indemnity function that minimizes the policyholder's risk measure of their end-of-period risk exposure (Theorem \ref{ThStep1}). Second, using the resulting optimal indemnity, we find the pricing distortion function that maximizes the insurer's expected profit (Corollary \ref{ThStep2}). We show that, in the first step, the optimal indemnity function exhibits a layer-type structure, determined by the interplay between the policyholder's distortion function and the insurer's pricing distortion function. Specifically, the optimal indemnity provides full indemnification over any loss layer on which the policyholder is more pessimistic than  the insurer's pricing functional about tail losses. When the policyholder is less pessimistic than the insurer's pricing functional about tail losses, no indemnification is offered. Finally, when the policyholder and the insurer's pricing functional are equally pessimistic about tail losses, the marginal indemnity may take an arbitrary shape, within the global $1$-Lispchitz (i.e., comonotonicity) constraints. In the second step, the optimal pricing distortion is determined by the policyholder's degree of (weak) risk aversion, that is, whether the policyholder's distortion function is above or below the identity function. In equilibrium, the insurer selects a pricing distortion that is aligned with the policyholder's risk perception, as encoded by their distortion function, in the sense that prices never exceed the policyholder's marginal willingness to insure tail losses. Moreover, we show that both the insurance coverage and the insurer's expected profit increase with the policyholder's degree of strong or weak risk aversion.

\medskip

We also analyze the Pareto efficiency of the Stackelberg equilibrium contracts. Our results show that any Stackelberg equilibrium contract is Pareto optimal and makes the policyholder indifferent between participating and not participating in the insurance market.  Moreover, any Pareto-optimal contract in which the policyholder is indifferent between participation and non-participation can be obtained at a Stackelberg equilibrium. These findings echo similar results obtained by \cite{BOONEN2023382} and \cite{GhossoubZhu2024}, and they highlight a well-known fundamental phenomenon that occurs in monopoly markets, whereby all consumer surplus is extracted by the monopoly.

\medskip

When the policyholder is either weakly risk averse (with a distortion function that lies above the identity function) or strongly risk averse (with a concave distortion function), we show that the optimal contract provides full insurance, and the optimal pricing distortion function coincides with that of the policyholder. This recovers the result of \cite{cheung2019risk} in the case of concave distortion function. For policyholders who minimize VaR at a confidence level $\alpha \in (0,1)$, the optimal coverage includes an upper limit. Specifically, the contract provides full insurance for losses below the $\alpha$-quantile while leaving the upper tail uninsured. This recovers another result of \cite{cheung2019risk}. However, in contrast, we show that for individuals with inverse-S shaped distortion functions, reflecting strong sensitivity to extreme losses, the optimal indemnity function takes the form of a deductible contract, whereby extreme losses are fully transferred to the insurer.

\medskip

The remainder of this paper is organized as follows. Section~\ref{SecProbFormDis} introduces the insurance market setting, including the risk preferences of the agents, the market mechanisms and resulting insurance contracts, and the (Stackelberg, or Bowley) equilibrium concept. Section~\ref{sec:SEcharacterization} provides a characterization of Stackelberg equilibria in our setting. Section~\ref{SecPO} examines the Pareto efficiency of the equilibrium contracts, and provides a version of the two welfare theorems. Section~\ref{SecExa} presents several examples of interest, for specific types of policyholders. Finally, Section~\ref{SecConDis} concludes. Proofs and related analysis are given in the \hyperlink{LinkToAppendix}{Appendices}.

%============================================
%============================================
%============================================
\bigskip 
\section{Problem Formulation} 
\label{SecProbFormDis}

Let $B(\mathcal{F})$ denote the set of all bounded random variables on a given non-atomic probability space $(\Omega,\mathcal{F},\mathbb{P})$. An individual, or potential policyholder, is subject to an insurable loss, which we model as a random variable $X \in B^+(\mathcal{F})$, the positive cone of $B(\mathcal{F})$, with range $[0,M]$. A positive realization of $X$ is  seen as a loss. We denote by $F_X$ the cumulative distribution function of $X$, and by $F_X^{-1}$ the left-continuous inverse of $F_X$, i.e., the quantile of $X$, defined as:
$$
F_X^{-1}(t) =\inf\left\{z\in\bbR^+ \, \middle| \, F_X(z)\ge t\right\},\ \forall t\in[0,1].
$$

\noindent Let $\mathcal{Q}$ denote the set of quantile functions of random variables in $B(\mathcal{F})$. That is,
$$
\mathcal{Q} =\left\{q:(0,1)\to\bbR^+\,\middle|\,q \mbox{ is non-decreasing and left-continuous}\right\}.
$$

%=========================
\medskip
\subsection{Preferences}
\begin{definition}
\label{defChoquetExpression}
A distortion risk measure $\rho$  on $(\Omega,\mathcal{F},\mathbb{P})$, is defined as the Choquet integral with respect to the distorted probability measure $T \circ \p$. Namely, for any $Y \in B(\cF)$, 
\begin{align*}
\rho(Y)
=\int\, Y\,\mathrm{d}T\circ\mathbb{P}
= \int_0^{+\infty}T(\mathbb{P}(Y\geq y))\,\mathrm{d}y + \int_{-\infty}^0 \left[T(\mathbb{P}(Y\geq y))-1\right]\mathrm{d}y,
\end{align*}

\noindent where $T:[0,1]\rightarrow[0,1]$ is a distortion function, that is, a non-decreasing differentiable mapping, satisfying $T(0)=0$ and $T(1)=1$.
\end{definition}

The conjugate of the distortion function $T$, is given by: $\tilde T (t)= 1 - T(1-t)$, for all $t \in [0,1]$. It is easy to verify that $\tilde T$ is a distortion function. 

\medskip

The policyholder's preference over $B(\mathcal F)$ is assumed to admit a distortion
risk measure representation, associated with a distortion function $T$. Specifically, for any
$Z \in B(\mathcal F)$, the policyholder evaluates risk according to
\begin{equation}\label{RiskPreference}
\rho^{Pol}(Z) := \int Z \, d \,T \circ \mathbb P.
\end{equation}

\noindent The induced preference relation $\succcurlyeq$ on $B(\mathcal F)$ is defined by
$ Z_1 \succcurlyeq Z_2$, if and only if, $\rho^{\mathrm{Pol}}(Z_1) \le \rho^{\mathrm{Pol}}(Z_2)$. 
Indifference is defined by $Z_1 \sim Z_2$, whenever
$Z_1 \succcurlyeq Z_2$, and $Z_2 \succcurlyeq Z_1$.

%=========================
\medskip
\subsection{Risk Aversion}
Given two random variables $Z_1, Z_2 \in B (\cF)$, we say that $Z_1$ dominates $Z_2$ in \textit{Second-Order Stochastic Dominance} (SSD), written $Z_1 \ssd Z_2$, if 
$$
\displaystyle \int_{-\infty}^{x}F_{Z_1}(t) \ dt \leq \int_{-\infty}^{x} F_{Z_2}(t) \ dt, \ \text{for all $x \in \mathbb{R}$}.
$$

\noindent Moreover, $Z_2$ is said to be a \textit{Mean-Preserving Increase in Risk} (MPIR) of $Z_1$ if $\E[Z_1] = \E[Z_2]$ and $Z_1 \ssd Z_2$. We next discuss weak and strong risk aversion of a preference $\succcurlyeq$ over $B(\cF)$. 

\medskip

\begin{definition}
\label{RiskAversion} 
A preference $\succcurlyeq$ over $B(\mathcal{F})$ is said to be weakly risk averse, if $\E[Z] \succcurlyeq Z$ for all $Z \in B(\mathcal{F})$. 
\end{definition}

\medskip

\begin{definition}
A preference $\succcurlyeq$ over $B(\mathcal{F})$ is said to be strongly risk averse if it ranks any random variable above all of its mean-preserving increases in risk. That is, $\succcurlyeq$ is strongly risk averse if $Z_1 \succcurlyeq Z_2$, for all $ Z_1, Z_2 \in B(\mathcal{F})$, such that $Z_2$ is a MPIR of $Z_1$.
\end{definition}

Strong risk aversion implies weak risk aversion, since any $Z\in B(\mathcal F)$
is a mean-preserving increase in risk of $\E[Z]$.
Moreover, by the positive homogeneity of the Choquet integral, it follows that, 
$$\E[Z] \succcurlyeq Z \Longleftrightarrow \E[Z] \leq \rho^{Pol}(Z), \ \forall Z \in B(\mathcal{F}).$$

\medskip

By a classical result (see, e.g., \cite{Yaari87}), risk aversion under distortion risk measures is characterized by properties of the distortion function $T$. In particular, the policyholder is weakly risk averse if and only if $T(t) \geq t$ for all $t\in[0,1]$. Moreover, the policyholder is strongly risk averse if and only if the distortion function $T$ is concave.

\medskip

Recall that for a given preference relation $\succcurlyeq$, the certainty equivalent of $Z \in B(\mathcal{F})$ is the constant $\text{CE}^\succcurlyeq(Z) \in \mathbb{R}$ such that 
$$Z \sim \text{CE}^\succcurlyeq(Z),$$

\noindent and the risk premium associated with $Z \in B(\mathcal{F})$ is defined as
$$\Delta^\succcurlyeq(Z) := \E[Z] - \text{CE}^\succcurlyeq(Z).$$

\medskip

Following \cite{ChewKarniSafra1987}, \cite{quiggin93}, and \cite{GhossoubHe2021}, we define comparative notions of risk aversion below.

\medskip

\begin{definition}[Comparative risk aversion]
\label{CompRiskAversion}
Consider two preference relations $\succcurlyeq$ and $\succcurlyeq^*$ over $B(\mathcal{F})$:
\begin{enumerate} 
\item  $\succcurlyeq^*$ is said to be {more weakly risk averse} than $\succcurlyeq$ if, for any $Z \in B(\mathcal{F})$, $\Delta^{\succcurlyeq^*}(Z) \geq \Delta^\succcurlyeq(Z)$.

\medskip

\item  $\succcurlyeq^*$ is said to be {more strongly risk averse} than $\succcurlyeq$ if $Z_1 \succcurlyeq^* Z_2$ for any $Z_1, Z_2 \in B(\mathcal{F})$ such that:
\begin{enumerate}
\item $Z_1 \sim Z_2$.
\medskip
\item There exists $z_0 \in \mathbb{R}$ with $F_{Z_2}(z) \geq F_{Z_1}(z)$ for all $z < z_0$, and $F_{Z_2}(z) \leq F_{Z_1}(z)$ for all $z \geq z_0$.
\end{enumerate} 
\end{enumerate} 
\end{definition}

By a classical result (e.g., \cite{ChewKarniSafra1987}), we obtain the following characterization of comparative weak and strong risk aversion for distortion risk measures. 

\medskip

\begin{proposition}\label{CompT}
Consider two policyholders whose preferences $\succcurlyeq$ and $\succcurlyeq^*$ over $B(\mathcal{F})$ admit representations by distortion risk measures $\rho^{Pol}$ and ${\rho^{*}}^{Pol}$ respectively. Let $T$ and $T^*$ denote the respective distortion functions of each policyholder. Then the following holds: 
\begin{enumerate}
\item The second policyholder is more weakly risk averse than the first policyholder if and only if $T^*(t) \geq T(t)$, for all $t \in [0,1]$.

\item The second policyholder is more strongly risk averse than the first policyholder if and only if $T^*$ is a concave transformation of $T$. That is, there exists an increasing and concave function $g: [0,1] \to [0,1]$, satisfying $g(0) = 0$, $g(1) = 1$, and $T^*(t) = g\left(T(t)\right)$ for all $t \in [0,1]$.
\end{enumerate}
\end{proposition}

\medskip

\begin{remark}\label{RemComp}
If $g$ is an increasing and concave function on $[0,1]$ such that $g(0)=0$ and $g(1)=1$, then $g(t) \geq t$, for all $t \in [0,1]$. Consequently, if $T^* \equiv  g \circ T$, we obtain that $T^*(t) \geq T(t)$ for all $t \in [0,1]$. That is, if $T^*$  is more strongly risk averse than $T$, then $T^*$ is also more weakly risk averse than $T$.
\end{remark}

%=====================
\medskip
\subsection{Market Mechanisms and Contracts}
The market allows the policyholder to cede part of the loss $X$ to an insurer, in exchange for a premium payment. We assume that the market only offers indemnities in the set of \textit{ex ante} admissible indemnity schedules $\mathcal{I}_L$ defined below. 
\begin{equation*}
\mathcal{I}_L = \Big\{I: [0,M]\rightarrow [0,M] \ \Big\vert \ I(0)=0, \ 0\leq I(x_1)- I(x_2)\leq x_1-x_2, \forall\, x_2\leq x_1 \in [0,M]\Big\}.
\end{equation*}

That is, $\cI_L$ is the set of $1$-Lipschitz functions satisfying the so-called \textit{no-sabotage condition} of \cite{carlierdana2003}, so as to rule \textit{ex post} moral hazard that could arise from misreporting of the actual realized loss.

\medskip

\begin{definition}
An insurance contract is a pair $(I,\pi)$, where $I \in \mathcal{I}_L$ is an indemnity function and $\pi\in\mathbb{R}$ is the premium paid by the policyholder for coverage $I$.
\end{definition}

\medskip

\begin{assumption}
The insurer is assumed to price insurance using a distortion premium principle of the form:
\begin{align}\label{Premium}
\Pi_g\left(I(X)\right)
:= \int I(X) \, d g \circ \mathbb{P}, \  \forall \, I \in \mathcal{I}_L,
\end{align}

\noindent for some distortion function $g$, which we hereafter refer to as the pricing distortion used by the insurer.
\end{assumption}

\medskip

\begin{definition}
\label{DefMechCont}
A market mechanism is a pair $(I,g)$, where $I \in \mathcal{I}_L$ is an indemnity function and $g$ is a pricing distortion function.
\end{definition}

\medskip

A market mechanism $(I,g)$ induces an insurance contract of the form $\left(I,\Pi_g \big( I(X) \big)\right)$, where the premium is computed using the pricing distortion $g$. The insurer's end-of-period profit is therefore given by $\Pi_g \big( I(X) \big) - I(X)$, while the policyholder's end-of-period risk exposure is given by $X-I(X) + \Pi_g \big( I(X) \big)$. Accordingly, the insurer's resulting expected profit is:
\begin{align}\label{InDis}
V^{In}(I,g) = \Pi_g \big( I(X) \big) -\mathbb{E}\left[I(X)\right],
\end{align}

\noindent and the policyholder evaluates this risk exposure using:
\begin{align}\label{PolDis1}
\rho^{Pol}(I,g) 
= \rho^{Pol} \Big( X - I(X) + \Pi_g \big( I(X) \big ) \Big).
\end{align}

Letting $R(X) := X - I(X)$ denote the retention function, i.e., the part of the loss $X$ that is retained by the policyholder, and by translation invariance of $\rho^{Pol}$, \eqref{PolDis1} reduces to:
\begin{equation}
\label{PolDis}
\rho^{Pol}(I,g) 
=\rho^{Pol}\big( R(X)  \big)  + \Pi_g \big( I(X) \big).
\end{equation}

\medskip

\begin{remark}
\label{re:Rfeasible}
\label{re:IandRcomonotone}
An indemnity function $I$ belongs to $\mathcal I_L$ if and only if the corresponding retention function $R(X)=X-I(X)$ belongs to $\mathcal I_L$.  Moreover, since $I$ is non-decreasing, the random variables $I(X)$ and $R(X)$ are comonotonic \footnote{Two random variables $X, Y \in B(\mathcal{F})$ are said to be comonotone if $(X(\omega_1) - X(\omega_2))(Y(\omega_1) -Y(\omega_2)) \geq 0$, for all $\omega_1,\omega_2\in \Omega$.}.
\end{remark}

%============================================
\medskip
\subsection{A Sequential-Move Game}
The insurance market is modeled as a sequential move game, in which the insurer, having the first-mover advantage, starts by selecting a pricing distortion function $g$. Given that choice, the policyholder then selects an indemnity function that minimizes their risk exposure $\rho^{Pol}(I,g)$. Anticipating the policyholder's optimal indemnity choice as a function of the selected pricing distortion function $g$, the insurer selects the optimal distortion function $g^*$ that maximizes their expected profit $V^{In}(I,g)$. The equilibrium concept that is best suited for this sequential game is the Stackelberg equilibrium. 

\medskip

\begin{definition}\label{BO}
A given market mechanism $( I^*,g^*) $ is said to be a Stackelberg Equilibrium (SE), if 
\begin{enumerate}
\item $I^* \in \underset{I \in \mathcal{I}_L} {\arg\min} \ \rho^{Pol}( I,{g^*})$, and
\medskip
\item $V^{In}(I^*,{g^*}) \geq V^{In}( I,g)$, for all $(I,g)$ such that $ I\in \underset{\bar I \in \mathcal{I}_L}{\arg\min} \ \rho^{Pol}( \bar I,g)$.
\end{enumerate}
\end{definition}

Definition \ref{BO} suggests that Stackelberg equilibria can be characterized through a two-step procedure. In the first step, for a fixed pricing distortion function $g$, the policyholder chooses an indemnity function that minimizes their risk exposure. This problem will be referred to as the \emph{policyholder's problem}. In the second step, anticipating the policyholder's optimal response as a function of $g$, the insurer selects a pricing distortion function $g^*$ that maximizes expected profit. This problem will be referred to as the \emph{insurer's problem}.

\medskip

For a given insurance contract $(I,\pi) \in  \mathcal{I}_L \times \mathbb{R}$, the policyholder's risk exposure and the insurer's expected profit can be written as 
$$
\rho^{Pol}(I,\pi) = \rho^{Pol}\left(X-I(X)+\pi\right)
\ \hbox{ and } \ 
V^{In}(I,\pi)=\pi -\mathbb{E}\left[I(X)\right].
$$

\medskip

\begin{definition}\label{def:IR}
An insurance contract $(I^*,\pi^*) \in  \mathcal{I}_L \times \mathbb{R}$ is said to be individually rational, IR, if
$$
\rho^{Pol}(I^*,\pi^*) \leq \rho^{Pol}(0,0)
\ \hbox{ and } \ 
V^{In}(I^*,\pi^*) \geq V^{In}(0,0).
$$
\end{definition}

Definition \ref{def:IR} states that an insurance contract $(I^*,\pi^*) \in  \mathcal{I}_L \times \mathbb{R}$ is individually rational, if it incentivizes the policyholder and the monopolist insurer to participate in the market. 

\medskip

\begin{definition}\label{def:PO}
An insurance contract $(I^*,\pi^*) \in  \mathcal{I}_L \times \mathbb{R}$ is said to be Pareto optimal, PO, if there does not exist another contract  $(I,\pi) \in \mathcal{I}_L \times \mathbb{R}$ such that 
$$
\rho^{Pol}(I,\pi) \leq \rho^{Pol}(I^*,\pi^*)
\ \hbox{ and } \ 
V^{In}(I,\pi) \geq V^{In}(I^*,\pi^*),
$$

\noindent with at least one strict inequality. 
\end{definition}

Definition \ref{def:PO} states that an insurance contract $(I^*,\pi^*) \in  \mathcal{I}_L \times \mathbb{R}$ is Pareto optimal if there is no alternative contract that weakly reduces the policyholder's risk exposure and weakly increases the insurer's profit, with at least one of these improvements being strict.

%============================================
%============================================
%============================================
\bigskip
\section{Characterization of Stackelberg Equilibria} 
\label{sec:SEcharacterization}

In this Section, we aim to characterize Stackelberg equilibria through a two-step procedure as suggested in Definition \ref{BO}. Specifically, in Subsection \ref{SecPH}, we study the policyholder's problem, which constitutes the first step in determining Stackelberg equilibria. Then, in Subsection \ref{SecIN}, we consider the second step by addressing the insurer's problem.

\medskip

\subsection{The Policyholder's Problem}
\label{SecPH}
For a given choice of pricing distortion function $g$, the policyholder chooses an indemnity $I \in \cI_L$ to minimize risk exposure: 
\begin{align}
\label{prob:Pol1}
\min\limits_{I\in \mathcal{I}_L} \, \rho^{Pol}(I,g).
\end{align}

To analyze the policyholder's problem given in \eqref{prob:Pol1}, we impose the following assumption on the loss distribution that will allow us to reformulate the problem in terms of quantile functions.

\medskip

\begin{assumption}\label{loss}
The cumulative distribution function $F_X$ is strictly increasing. 
\end{assumption} 

It follows from Assumption \ref{loss} that $F_X$ is differentiable almost everywhere. Moreover, by \citet[Lemma A.25]{FollmerSchied2016}, Assumption \ref{loss} also guarantees that $U := F_X(X)$ is uniformly distributed on $(0,1)$, and that $X=F_X^{-1}(U), \ \p$-a.s. 

\medskip

\begin{remark}
For each $I \in \mathcal{I}_L$, $I(X)$ and $X-I(X)$ have strictly increasing cumulative distribution functions by Assumption \ref{loss}. Consequently, their quantile functions are strictly increasing, left-continuous, and differentiable a.e.\ on $[0,1]$.
\end{remark}

\medskip

For any $Z \in B^+(\mathcal{F})$ whose quantile function $F_Z^{-1}$ is differentiable almost everywhere, the policyholder's risk measure admits the following representation:
\begin{align*}
\rho^{Pol}(Z)
&=\int\, Z \,\mathrm{d}T\circ\mathbb{P} 
=\int \left(F_Z^{-1}\right)^\prime(U) \, T(1-U)\,\mathrm{d}\p.
\end{align*}

\noindent Similarly, the premium evaluated under the given pricing distortion $g$ can be written as
\begin{align*}
\Pi_g\left(Z\right)
=\int Z \, d g \circ \mathbb{P}
=\int \left(F_Z^{-1}\right)^\prime(U) \, g(1-U)\,\mathrm{d}\p.
\end{align*}

\noindent Applying these formulas to the retention $R(X)$ and the indemnity $I(X)$ respectively, and using expression \eqref{PolDis} of  $\rho^{\mathrm{Pol}}(I,g)$, we obtain
\begin{equation}
\label{eq:rho_quantile}
\rho^{Pol}(I,g)
=\int \left(F_{R(X)}^{-1}\right)^\prime(U) \ T(1-U)\,\mathrm{d}\p
+
\int \left(F_{I(X)}^{-1}\right)^\prime(U) \ g(1-U)\,\mathrm{d}\p. 
\end{equation}

\medskip

\begin{remark}
\label{re:quantilecoadd}
Since $X = I(X) + R(X)$, we can write $F_X^{-1} = F_{I(X)}^{-1} + F_{R(X)}^{-1}$, by comonotonic additivity of the quantile function. 
\end{remark}

As a result of the above remark, \eqref{eq:rho_quantile} can be rewritten as follows:
\begin{align*}
\rho^{Pol}(I,g)
&=\int_0^1 \left(F_{R(X)}^{-1}\right)^\prime(t) \ \left[T(1-t) - g(1-t)\right]\mathrm{d}t
+
\int_0^1 \left(F_{X}^{-1}\right)^\prime(t)
\ g(1-t)\mathrm{d}t.
\end{align*}

\medskip

Let $\cQ_L$ be the set of admissible quantile functions, defined as
$$
\mathcal{Q}_L=\left\{q \in \mathcal{Q}\, \Big |\,q(0)=0, \ 0\leq q^\prime(t)\leq \left(F_X^{-1}\right)^\prime(t)\right\}.
$$

\noindent Each $q \in \mathcal Q_L$ corresponds to the quantile of the retention random variable $R(X)=X-I(X)$ for some admissible $I \in \mathcal I_L$. Hence, for a given $q \in \mathcal{Q}_L$, we have:
\begin{align}
\label{eq:pol_quantile_rep}
\rho^{Pol}(q,g)
&=\int_0^1 q^\prime(t) \, \left[T(1-t) - g(1-t)\right]\,\mathrm{d}t + \int_0^1 \left(F_{X}^{-1}\right)^\prime(t) \, g(1-t)\,\mathrm{d}t.
\end{align}

\medskip

\noindent Reformulating the policyholder's problem \eqref{prob:Pol1} in quantile form, we obtain:
\begin{align}
\label{DisMin}
\min\limits_{q\in \mathcal{Q}_L} \rho^{Pol}( q,g).
\end{align}    

\medskip

\begin{lemma}
\label{le:optimality_qg_iff_Ig} 
For a given pricing distortion function $g$, the feasible quantile $q_{g}(t) \in \cQ_L$ is optimal for \eqref{DisMin} if and only if the indemnity $I_g(x)= x - q_g \left(F_X(x)\right)$ is optimal for \eqref{prob:Pol1}.
\end{lemma}

\begin{proof}
The proof can be found in Appendix \ref{App:le:optimality_qg_iff_Ig}. 
\end{proof}

\medskip

\begin{lemma}
\label{le:qg_characterization}
For a given pricing distortion function $g$, a quantile function $q_g$ is optimal for \eqref{DisMin} if and only if 
\begin{equation}
\label{qg}
\left(q_{{g}}\right)^\prime(t)=
\left\{
\begin{array}[c]{ll}%
0, &g(1-t)<T(1-t),\vspace{0.2cm}\\
\phi_{ g}(t), & g(1-t)=T(1-t),\vspace{0.2cm}\\
\left(F_X^{-1}\right)^\prime(t), & g(1-t)>T(1-t),
\end{array}
\right.
\end{equation}

\medskip

\noindent where $\phi_g(t)\in\left[0,\left(F_X^{-1}\right)^\prime(t)\right]$, for almost every $t \in [0,1]$ such that $g(1-t)=T(1-t)$. 
\end{lemma}

\begin{proof}
The proof can be found in Appendix \ref{App:le:qg_characterization}. 
\end{proof}

\medskip

\begin{theorem} 
\label{ThStep1}
For a given pricing distortion function $g$, an indemnity function $I_g$ is optimal for the policyholder's problem given in \eqref{prob:Pol1} if and only if it is of the form $I_g(x) = \int_0^x \kappa(y) \, dy$, for all $x \in [0,M]$, with $\kappa: [0,M] \to [0,1]$ satisfying the following:
\begin{equation*}
\kappa(y) =
\left\{
\begin{array}[c]{ll}%
1, &  g\left(\p[X > y]\right) <  T\left(\p[X > y]\right),\vspace{0.2cm}\\
1 - \phi_{g}(F_X(y)) \, f_X(y) , &  g\left(\p[X > y]\right) =  T\left(\p[X > y]\right),\vspace{0.2cm}\\
0, &  g\left(\p[X > y]\right) >  T\left(\p[X > y]\right),
\end{array}
\right.
\end{equation*}

\medskip

\noindent where $\phi_g(t)\in\left[0,\left(F_X^{-1}\right)^\prime(t)\right]$, for almost every $t \in [0,1]$ such that $ g(1-t) =  T(1-t)$, and $f_X$ denotes the probability density function of the loss $X$.
\end{theorem}

\begin{proof}
The proof follows immediately from Lemmata \ref{le:optimality_qg_iff_Ig} and \ref{le:qg_characterization}. 
\end{proof}

\medskip

Theorem \ref{ThStep1} characterizes the set of optimal indemnity functions for a given pricing distortion function $g$, in terms of the policyholder's marginal indemnification. Specifically, full indemnification is optimal when the insurer's pricing distortion assigns less weight to tail probabilities than the policyholder's distortion function $T$. The policyholder retains the entire loss if the insurer's pricing distortion overweights the tail probability compared to the policyholder's distortion. Finally, when the policyholder's distortion is equal to the insurer's pricing distortion at a given tail probability, the policyholder may receive partial coverage, as long as feasibility is maintained. 

\medskip

This structural characterization is consistent with \cite{ASSA201570}, who considers a reinsurance problem in which the premium principle is fixed and distortion based, and characterizes the optimal contract of the policyholder. In contrast, our result arises as the policyholder's best response within a Stackelberg framework, where the pricing distortion is a strategic choice of the insurer. Moreover, in the absence of strategic interaction and when the policyholder's problem only is considered, our model reduces to the setting analyzed in \cite{ASSA201570}.

%============================================
\medskip
\subsection{The Insurer's Problem}
\label{SecIN}
The optimal indemnity characterized in Theorem \ref{ThStep1} is not unique. The insurer's objective is to identify a market mechanism $(I^*_{g^*},g^*)$ that solves the following problem:
\begin{align}
\label{IN2}
\max_g \, V^{In}(I_g,g),
\ \text{such that} \
I_g \in \arg\min_{I\in\mathcal I_L} \, \rho^{Pol}(I,g).
\end{align}

\medskip

\begin{lemma}
The market mechanism $(I^*_{g^*},g^*)$ is a Stackelberg equilibrium if and only if it is optimal for the insurer's problem in \eqref{IN2}.
\end{lemma}

\begin{proof}
The proof follows immediately from Definition \ref{BO}.
\end{proof}

The insurer's problem \eqref{IN2} can be reformulated using quantile functions, similarly to the policyholder's problem analyzed in Subsection \ref{SecPH}. Consider a market mechanism $(I,g)$, using Remark \ref{re:quantilecoadd}, the premium can be written as
$$
\Pi_g \big( I(X) \big)=
\int_0^1 \left[\left(F_{X}^{-1}\right)^\prime(t) - \left(F_{R(X)}^{-1}\right)^\prime(t) 
\right]
\, g(1-t)\,\mathrm{d}t.
$$

\noindent Moreover, 
\begin{align*}
\mathbb{E}\left[I(X)\right] 
&=\int_0^1 \left[F_{X}^{-1}(t) - F_{R(X)}^{-1}(t) \right] \,\mathrm{d}t.
\end{align*}

\noindent Substituting the quantile representations of the premium and the expected indemnity into $V^{In}( I_g,g)$ in \eqref{InDis}, where $R_g$ denotes the retention function associated with $I_g$,   yields:
\begin{align*}
V^{In}( I_g,g)
&=\int_0^1 \left[\left(F_{X}^{-1}\right)^\prime(t) - \left(F_{R_g(X)}^{-1}\right)^\prime(t) 
\right]
\, g(1-t)\,\mathrm{d}t 
-\int_0^1 \left[F_{X}^{-1}(t) - F_{R_g(X)}^{-1}(t) 
\right]
\,\mathrm{d}t.
\end{align*}

\noindent For a given quantile $q \in \mathcal{Q}_L$, this expression reduces to:
\begin{align*}
V^{In}(q,g)
&=\int_0^1 \left[\left(F_{X}^{-1}\right)^\prime(t) - q^\prime(t) 
\right]
\, g(1-t)\,\mathrm{d}t 
-\int_0^1 \left[F_{X}^{-1}(t) - q(t) 
\right]
\,\mathrm{d}t.
\end{align*}

\medskip

\noindent Hence, the insurer's problem can be equivalently written in quantile form as follows:
\begin{align}\label{IN3}
\max\limits_{g} \, V^{In}( q_g,g), \ \text{such that $q_g \in \arg\min_{q \in \mathcal Q_L} \rho^{Pol}(q,g)$}. 
\end{align}

\medskip

The following result suggests that solving the insurer's problem \eqref{IN2} is equivalent to solving \eqref{IN3}, providing a one-to-one correspondence between optimal solutions in the quantile space and optimal mechanisms.

\medskip

\begin{lemma}
\label{LemEq2}
$(q_g,g)$ is optimal for \eqref{IN3} if and only if $(I_g,g)$ is optimal for \eqref{IN2}, where $I_g(x) = x - q_g\left(F_X(x)\right)$, for all $x \in [0,M]$.
\end{lemma}

\begin{proof}
The proof is similar to that of Lemma \ref{le:optimality_qg_iff_Ig}.
\end{proof}

\medskip

\begin{theorem}
\label{OptDis} 
$(q^*_{g^*},{g}^*)$ is optimal for \eqref{IN3} if and only if the following two conditions hold: 
\begin{enumerate}
\item  The optimal pricing distortion $g^*$ satisfies $g^*(t)=1-\tilde{g} ^*(1-t)$ for all $t \in [0,1]$, where the optimal pricing conjugate $\tilde g ^*$ is given by:
\begin{equation}
\label{tilde}
\tilde g^*(t) = \left\{ 
\begin{array}[c]{ll}%
\tilde T(t), &\tilde T(t)<t,\vspace{0.3cm}\\
\in\left[\,\sup \left\{z <t;\ \tilde g^*(z) \right\}, \tilde T(t)\right], & \tilde T(t)\geq t,
\end{array}\right.  
\end{equation}

\medskip

\noindent and $\tilde T$ is the conjugate distortion function of $T$. 

\bigskip

\item The optimal quantile $q^*_{g^*}$ satisfies:
\begin{equation}
\label{qast}
\left(q^*_{{g}^*}\right)^\prime(t)=
\left\{
\begin{array}[c]{ll}%
0, &\tilde T(t)<t,\vspace{0.2cm}\\
\phi_{ g^*}(t)\, & t\in\{z; \ \tilde{g}^*(z)= \tilde T(z)=z\},\vspace{0.2cm}\\
\left(F_X^{-1}\right)^\prime(t), &t\in\{z; \ \tilde{g}^*(z)<\tilde T(z)=z\}\cup\{z; \ \tilde T(z)>z\},
\end{array}
\right.  
\end{equation}

\medskip

\noindent where $\phi_{g^*}(t)\in\left[0,\left(F_X^{-1}\right)^\prime(t)\right]$ for a.e.\ $t\in[0,1]$ such that $\tilde{g}^*(t)= \tilde T(t)=t.$
\end{enumerate}

\bigskip

\noindent Moreover, the insurer's expected profit under $(q^*_{g^*}, g^*)$ is given by:
\begin{equation} \label{eq:exp_profit}
V^{In}\left( q^*_{g^*}, g^*\right) = \int_0^1\left(F_X^{-1}\right)^\prime(t)\,(t-\tilde T(t))\One_{\{\tilde T(t)<t\}} \, \mathrm{d}t. 
\end{equation}
\end{theorem}

\begin{proof}
The proof can be found in Appendix \ref{AppProofThOptDis}.
\end{proof}

\medskip

Theorem \ref{OptDis} provides a complete characterization of the Stackelberg equilibrium by identifying the optimal pricing distortion $g^*$ and the induced optimal quantile $q^*_{g^*}$. The following corollary explicitly provides the characterization of Stackelberg equilibria in terms of market mechanisms of the form $(I_g, g)$. 

\medskip

\begin{corollary}
\label{ThStep2}
The market mechanism $(I^*_{g^*}, g^*)$ is optimal for \eqref{IN2} if and only if the following two conditions hold:
\medskip
\begin{enumerate}
\item The optimal pricing distortion $g^*$ is given by:
\begin{equation}
\label{OptG}
 g^*(t) =
\left\{
\begin{array}[c]{ll}%
T(t), & T(t)>t,\vspace{0.3cm}\\
\in\left[\,  T(t), \,\inf \left\{z >t :  g^*(z) \right\}\right], &  T(t)\leq t,
\end{array}
\right.  
\end{equation}

\medskip

\item The optimal indemnity satisfies $I^*_{g^*}(x) = \int_0^x \kappa^*(y) \, \mathrm{d}y$ for all $x \in [0,M]$, where $ \kappa^*: [0, M] \to [0, 1] $ is defined as follows:
\medskip
$$
\kappa^*(y) :=
\begin{cases}
1, & \mathbb{P}[X > y] < T(\mathbb{P}[X > y]), \\[0.3em]
1 - \phi_{g^*}(F_X(y)) \, f_X(y), & g^*(\mathbb{P}[X > y]) = T(\mathbb{P}[X > y]) = \mathbb{P}[X > y], \\[0.3em]
0, &  g^*(\mathbb{P}[X > y]) > T(\mathbb{P}[X > y]) = \mathbb{P}[X > y],\text{or } \mathbb{P}[X > y] > T(\mathbb{P}[X > y]),
\end{cases}
$$

\medskip

\noindent where $\phi_{g^*}(t) \in \left[0, \left(F_X^{-1}\right)'(t)\right]$, for a.e.\ $t \in [0,1]$ such that $g^*(1 - t) = T(1 - t) = 1 - t$.
\end{enumerate}
\end{corollary}

\begin{proof}
The proof follows immediately from Lemma \ref{LemEq2} and Theorem \ref{OptDis}.
\end{proof}

\medskip

\begin{theorem}\label{CompStatWRA}
Consider two policyholders whose respective distortion functions $T_1$ and $T_2$ satisfy $T_1(t) \leq T_2(t)$ for all $t \in [0,1]$. Let $(I_{g^*_1}^*, g_1^*)$ and $(I_{g^*_2}^*, g_2^*)$ denote the corresponding Stackelberg equilibria. Then, the following holds.
\begin{enumerate}
\item $I_{g^*_1}^*(x) \le I_{g^*_2}^*(x)$, for all $x \in [0,M]$. 
\item $V^{In}(I_{g^*_1}^*, g_1^*) \le V^{In}(I_{g^*_2}^*, g_2^*)$
\end{enumerate}
\end{theorem}

Theorem \ref{CompStatWRA} states that under Stackelberg equilibria, the optimal insurance coverage and the insurer's expected profit both increase as the policyholder becomes weakly more risk averse in the sense of Proposition \ref{CompT}. That is, the more weakly risk averse the policyholder is, the more coverage they receive, and the more profitable the insurer becomes. The following corollary shows that this result also holds under strong risk aversion. 

\medskip

\begin{corollary} 
\label{CompStatSRA}
Under Stackelberg equilibria, if the policyholder becomes more strongly risk averse, then both the optimal insurance coverage and the insurer’s expected profit increase.
\end{corollary}

\begin{proof}
The proof of this result follows immediately from Remark \ref{RemComp} and Theorem \ref{CompStatWRA}. 
\end{proof}

In contrast to this clear monotonic relationship with respect to risk aversion obtained in Theorem \ref{CompStatWRA} and Corollary \ref{CompStatSRA}, no such monotonicity can generally be established for the insurance coverage or the insurer's expected profit when the policyholder's risk distribution itself changes.

%============================================
%============================================
%============================================
\bigskip
\section{Pareto Efficiency of Stackelberg Equilibria}
\label{SecPO}
In this section, we examine whether the Stackelberg equilibria characterized in Section \ref{sec:SEcharacterization} lead to Pareto efficient contracts. We first characterize Pareto optimal insurance contracts and then establish their relationship with Stackelberg equilibria. The proofs of these results can be found in Appendix \ref{AppProofChaPo}, Appendix \ref{AppProofSeIsPo}, and Appendix \ref{AppProofPoIsSe}. 

\medskip

\begin{proposition}\label{ChaPo}
An insurance contract $(I^*,\pi^*)$ is Pareto optimal if and only if it solves the following problem:
\begin{equation}\label{ChaPoPro}
\min\limits_{(I,\pi) \in \mathcal{I}_L \times \mathbb{R}} \,\left\{\rho^{Pol}(I,\pi)-V^{In}(I,\pi)\right\}.
\end{equation}
\end{proposition}

\medskip

\begin{proposition}\label{SeIsPo}
Let $(I^*_{g^*}, g^*)$ be a Stackelberg equilibrium characterized in Corollary \ref{ThStep2}. The induced insurance contract $(I^*_{g^*}, \pi^*_{g^*})$ is individually rational and Pareto optimal. Moreover, $\rho^{Pol}(I^*_{g^*},\pi^*_{g^*}) = \rho^{Pol}(0,0)$.
\end{proposition}

Proposition \ref{SeIsPo} establishes a first welfare theorem for the sequential-move monopolistic market considered in this paper. Specifically, Stackelberg equilibria lead to Pareto optimal insurance contracts. Moreover, in equilibrium, the policyholder is indifferent between participating in the market and not participating, and the monopolist insurer extracts all surplus. In contrast, not every Pareto optimal contract arises from a Stackelberg equilibrium mechanism. The following result discusses the conditions under which Pareto optimal contracts are induced by Stackelberg equilibria. 

\medskip

\begin{proposition}
\label{PoIsSe}
Consider a Pareto optimal contract $(I^*, \pi^*)$ satisfying $\rho^{\text{Pol}}(I^*, \pi^*) = \rho^{\text{Pol}}(0, 0)$. Then there exists a Stackelberg equilibrium $(I^*, g^*_{I^*})$ that induces the Pareto optimal contract $(I^*, \pi^*)$, that is,
$$\Pi_{g^*_{I^*}}\left(I^*(X)\right) = \pi^*.$$
\end{proposition}

Proposition \ref{PoIsSe} states that the Pareto optimal contracts that leave the policyholder indifferent between participating and not participating in the market are induced by Stackelberg equilibrium mechanisms. 

\medskip

The results of this section are consistent with findings of the existing literature, e.g., \cite{BOONEN2023382}, \cite{GhossoubZhu2024}, in which the monopolistic market is modeled as a sequential-move game. Additionally, \cite{ANDRAOS2026103210} consider a generalized framework and obtain similar results for the special case of monopoly. Specifically, Stackelberg equilibria yield Pareto optimality without inducing any welfare gain to the policyholder. Conversely, only the Pareto optimal contracts that leave the policyholder indifferent between suffering loss and entering the market are induced from Stackelberg equilibrium mechanisms.

%============================================
%============================================
%============================================
\bigskip
\section{Examples}
\label{SecExa}
In this section, we assume that the policyholder evaluates risk using a specific class of risk measures, and we examine the resulting Stackelberg equilibria. As established in Lemma~\ref{le:qg_characterization} and Theorem~\ref{OptDis}, the solutions to both stages of the Stackelberg game are generally not unique. Nevertheless, we show that under certain assumptions, Stackelberg equilibria may take one of several familiar forms such as full insurance, a coverage limit insurance contract, or a deductible insurance contract.

\bigskip

\subsection{Optimality of Full Insurance}
\label{concave}
Suppose that the policyholder is weakly risk averse, that is, the distortion function satisfies 
$$
T(t) \geq t, \ \forall t \in [0,1].
$$

\noindent It follows from Corollary~\ref{ThStep2}, that the market mechanism $(I^*_{g^*}, g^*)$ described below is a Stackelberg equilibrium:
$$
g^*(t) = T(t), \ \forall \, t \in [0,1],
\ \ \hbox{and} \ \ 
I^*_{g^*}(x) = x, \ \forall \, x \in [0,M].
$$

\noindent That is, the insurer offers full coverage, and the pricing distortion function coincides with the distortion function that reflects the policyholder's risk aversion. This result also holds for a strongly risk averse policyholder, by Remark \ref{RemComp}.

This example covers many important risk preference models commonly used in practice, such as Tail Value-at-Risk (TVaR), which corresponds to a concave distortion function of the form 
$$
T(t) := \min\left(1, \frac{t}{1-\alpha}\right), \ \text{for a given $\alpha \in (0,1)$}.
$$

\noindent Moreover, under TVaR, the insurer's expected profit satisfies:
$$
V^{{In}}(I^*_{g^*},\, g^*)
= \int_0^1 \big(F_X^{-1}\big)^\prime(t)\,\big[t - \tilde{T}(t)\big]\, \One_{\{\tilde{T}(t) < t\}}\, \mathrm{d}t \geq 0,
$$

\noindent since the set $\{t\in[0,1]: \tilde T(t)< t\}$ has positive measure.

%============================================
\bigskip
\subsection{Optimality of Coverage Limit Contracts}
\label{ValueAtRisk}
In this example, we assume that the policyholder evaluates risk using Value-at-Risk (VaR) at confidence level $\alpha \in (0,1)$. This corresponds to a distortion risk measure with distortion function 
$$
T(t) := \One_{(1-\alpha,\,1]}(t), \ \forall t \in [0,1]. 
$$

It follows from Corollary \ref{ThStep2} that a Stackelberg equilibrium $(I^*_{g^*}, g^*)$ is characterized as follows. The optimal pricing distortions $g^*$ satisfies 
$$
g^*(t) =
\begin{cases}
1, & t > 1-\alpha,  \quad  \ \\[6pt]
\in \Big[\,0,\, \inf\big\{ z > t : g^*(z) \big\} \Big], & t \leq 1-\alpha, \quad 
\end{cases} \quad  \text{for a.e $t \in [0,1]$, }
$$

\noindent and the optimal indemnity function $I^*_{g^*}$ is given by
$$
I^*_{g^*}(x) =
\begin{cases}
x, & x < F_X^{-1}(\alpha), \quad \\[6pt]
F_X^{-1}(\alpha), & x \geq F_X^{-1}(\alpha), \quad 
\end{cases} \quad \forall x \in [0,M]. 
$$

In this case, the policyholder receives full coverage for losses below the VaR threshold $F_X^{-1}(\alpha)$. Moreover, coverage is capped at this threshold level, so that for losses exceeding $F_X^{-1}(\alpha)$, the policyholder retains the excess loss. The insurer's expected profit satisfies:
$$
V^{{In}}(I^*_{g^*},\, g^*)
= \int_0^1 \big(F_X^{-1}\big)^\prime(t)\,\big[t - \tilde{T}(t)\big]\, \One_{\{\tilde{T}(t) < t\}}\, \mathrm{d}t
= \int_0^\alpha t \, \cdot \big(F_X^{-1}\big)^\prime(t) \, \mathrm{d}t, 
$$

\noindent since $\tilde T(t)=\One_{[\alpha,1]}(t)$, for all $t \in [0,1]$. Additionally, note that since $F_X^{-1}$ is strictly increasing, it follows that for a fixed $x \in [0,M]$, the optimal indemnity function $I^*_{g^*}$ weakly increases with the VaR confidence level $\alpha$. Moreover, the insurer's expected profit increases with $\alpha$. 

\medskip

Hence, if the policyholder is a VaR minimizer, then the coverage limit contract is a Stackelberg equilibrium and depends on the chosen VaR confidence level $\alpha$. In addition, a more risk-averse policyholder (with a higher value for $\alpha$) receives greater coverage and generates higher expected profit for the insurer, consistently with Theorem \ref{CompStatWRA} and Corollary \ref{CompStatSRA}.

\medskip

The results of our examples so far, align with the findings of \citet{cheung2019risk}, where the policyholder is assumed to be either strongly risk averse or a VaR minimizer. Specifically, if the policyholder's preferences are represented by a DRM with a concave distortion function, the results of Example \ref{concave} align with  \citet[Theorem 3.1, case $\gamma=0$]{cheung2019risk}. On the other hand, if the distortion function corresponds to a VaR risk measure, then the results of Example \ref{ValueAtRisk} are consistent with \citet[Theorem 3.5, case $\gamma=0$]{cheung2019risk}.

%============================================
\bigskip
\subsection{Optimality of Deductible Contracts}
We now assume that the policyholder uses an inverse-S-shaped distortion function (ISSD), which is commonly used to study decision making under uncertainty (e.g., \cite{Tversky1992}). 

\medskip

\begin{definition}
A distortion function $T$ is said to be an inverse-S-shaped distortion function (ISSD) if it is twice-differentiable on $(0,1)$, and there exists  $t_0 \in (0,1)$ such that:
\begin{enumerate}
\item  $T^   {\prime} (t)$ is strictly deceasing on $(0,t_0)$, and
\medskip
\item  $T^ \prime(t)$ strictly increasing on $(t_0,1).$ 
\end{enumerate}

\medskip

\noindent Moreover, $\lim_{t\downarrow 0}T^\prime(t)>1,$ and $\lim_{t\uparrow 1}T^\prime(t)>1.$
\end{definition}

\medskip

In this case, for $Y \in B^+ (\cF)$, the expression of the distortion risk measure is given by
\begin{align*}
\rho(Y) 
= \int_0^{+\infty}T(\mathbb{P}(Y\geq y))\,\mathrm{d}y
=\int_0^{+\infty}\,y\,T^\prime(1-F_Y(y))\,\mathrm{d}F_Y(y).
\end{align*}

Assume that there exists $t_1\in(0,1)$ such that $T(t_1)=t_1$. That is, $t_1$ is the intersection point between the identity function and the policyholder's distortion function. We know from Corollary~\ref{ThStep2}, that the Stackelberg equilibrium $(I^*_{g^*}, g^*)$ is characterized as follows. The optimal pricing distortion $g^*$ satisfies:
$$
g^*(t) =
\begin{cases}
T(t), & t < t_1, \quad \\[6pt]
\in \big[\, T(t),\, \inf\{ z > t : g^*(z) \} \big], & t \ge t_1,\quad 
\end{cases} \quad  \text{for a.e. $t \in [0,1]$,}
$$

\noindent and the optimal indemnity function $I^*_{g^*}$ satisfies:
$$
I^*_{g^*}(x) =
\begin{cases}
0, & x \le F_X^{-1}(1 - t_1), \quad \\[6pt]
x - F_X^{-1}(1 - t_1), & x > F_X^{-1}(1 - t_1), \quad
\end{cases} \quad \forall x \in [0,M]. 
$$

In this case, the optimal indemnity fully covers losses above a fixed deductible $ F_X^{-1}(1 - t_1)$, which is fully determined by $t_1$. Hence, if the policyholder is more concerned about extreme losses, then the deductible contract is a Stackelberg equilibrium. Moreover, the insurer's expected profit is given by:
\begin{align*}
V^{In}\left( I^*_{{g}^*},{g}^*\right)
&=\int_0^{t_1}\left(F_X^{-1}\right)^\prime(1-t)\left(T(t)-t\right)\,\mathrm{d}t.
\end{align*}

We note that for a fixed $x \in [0,M]$, as $t_1$ increases, $ F_X^{-1}(1 - t_1)$ decreases, implying a greater insurance coverage. That is, the optimal indemnity $I^*_{g^*}$ increases with $t_1$. However, from the expression of the expected profit at equilibrium, we cannot conclude whether $V^{In}\left( I^*_{{g}^*},{g}^*\right)$ increases with $t_1$. This is because a higher value of $t_1$ implies a change in the distortion function $T$ on the interval $[0,t_1)$. Hence, to evaluate the impact of $t_1$ on the insurer's expected profit, we consider the following concrete numerical example, in which the distortion function takes the form given in \cite{Tversky1992}:
\begin{equation*}
T_\theta(t)
:=\frac{t^\theta}{\left(t^\theta+(1-t)^\theta\right)^{\frac{1}{\theta}}},
\end{equation*}

\noindent where $\theta \in [0.3,0.8]$. Additionally, we consider three cases for the distribution of the random loss $X$. 

\medskip

We first assume that the random loss $X$ follows a uniform distribution with cumulative distribution function (CDF) given by 
$$
F_X(x)=\frac{x}{M}, \ \forall x \in [0,M] \ \text{and $M = 10$}. 
$$

\medskip

Figure \ref{F1-a} plots the distortion functions $T_\theta$ for different values of $\theta$. The figure shows that as $\theta$ increases, $T_\theta$ approaches the identity function, and the intersection point $t_1$ shifts to larger values. Moreover, functions with smaller values of $\theta$ are more concave near $0$ and more convex near $1$, implying that the policyholder places greater emphasis on extreme losses. 

\medskip

Figure \ref{F1-b} depicts the insurer's expected profit under the Stackelberg optimal contract as a function of $t_1$. The result shows that the insurer's expected profit does not vary monotonically with $\theta$ or the intersection point $t_1$. Notably, it reaches a maximum around $t_1 \approx 0.32$.

\medskip

\begin{figure}[H] 
\begin{subfigure}{0.48\textwidth}
\centering
\includegraphics[height=8cm]{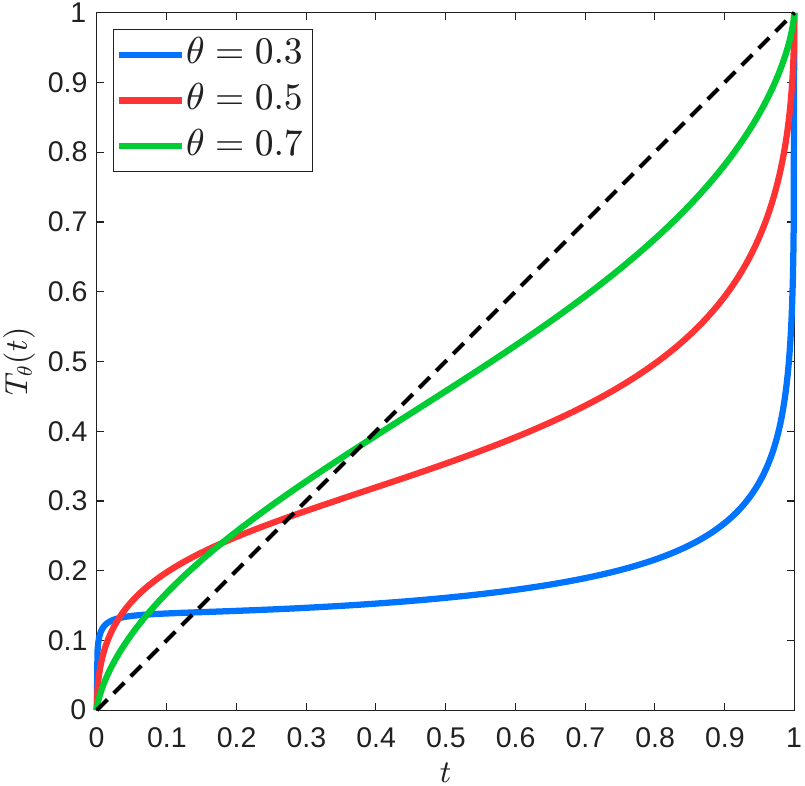}
\caption{Policyholder’s Distortion Function.}
\label{F1-a}
\end{subfigure}\hfill
\begin{subfigure}{0.48\textwidth}
\centering
\includegraphics[height=8cm]{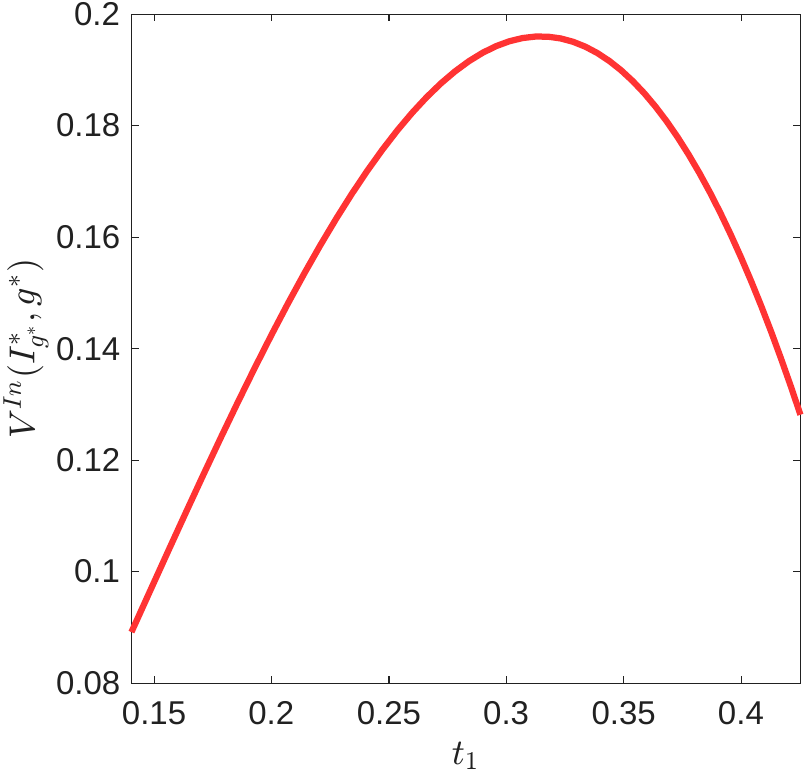}
\caption{Insurer’s Expected Profit Under SE.}
\label{F1-b}
\end{subfigure}
\caption{The case where $X$ follows a uniform distribution.}
\label{Fig:bowley}
\end{figure}

\medskip

Alternatively, assume now that the loss random variable $X$ follows a truncated exponential distribution, with CDF given by
$$
F_X(x) = \frac{1 - \exp(-\lambda x)}{1 - \exp(-\lambda M)}, \ \forall x \in [0, M], \ M=10, \ \text{and $\lambda >0$}. 
$$

\medskip

Figure \ref{F2-a} plots the CDF for different values of $\lambda$. It shows that as $\lambda$ increases, the CDF becomes steeper for small losses, implying that losses are more likely to take smaller values. Figure \ref{F2-b} plots the insurer's profit as a function of $t_1$ for different values of $\lambda$. The insurer's profit does not vary monotonically with the parameter $\theta$ or the intersection point $t_1$.

\medskip

The result also depends on the loss distribution parameter $\lambda$.  In particular, the insurer's profit reaches its maximum around $t_1 \approx 0.3$ when $\lambda = 0.1$. When $\lambda$ increases to $0.5$, meaning that the distribution becomes more concentrated on the left, the maximum shifts to about $t_1 \approx 0.22$, indicating that if the policyholder faces a lower probability of large losses, greater concern for extreme losses becomes more valuable to the insurer. When $\lambda$ increases further to $1$, implying that losses are even more likely to be small, the maximum shifts slightly further left, and the insurer's profit then decreases monotonically with $t_1$. This suggests that the insurer achieves the highest profit when the policyholder places more weight on extreme loss events.

\medskip

\begin{figure}[H]
\centering
\begin{subfigure}{0.48\textwidth}
\centering
\includegraphics[height=8cm]{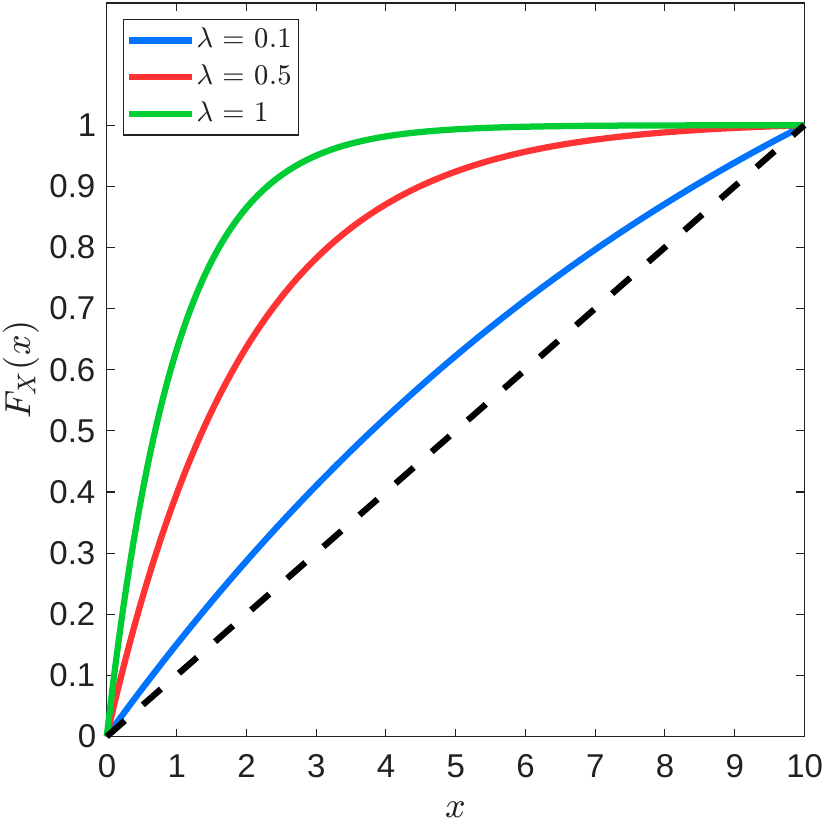}
\caption{Cumulative Distribution Function.}
\label{F2-a}
\end{subfigure}\hfill
\begin{subfigure}{0.48\textwidth}
\centering
\includegraphics[height=8cm]{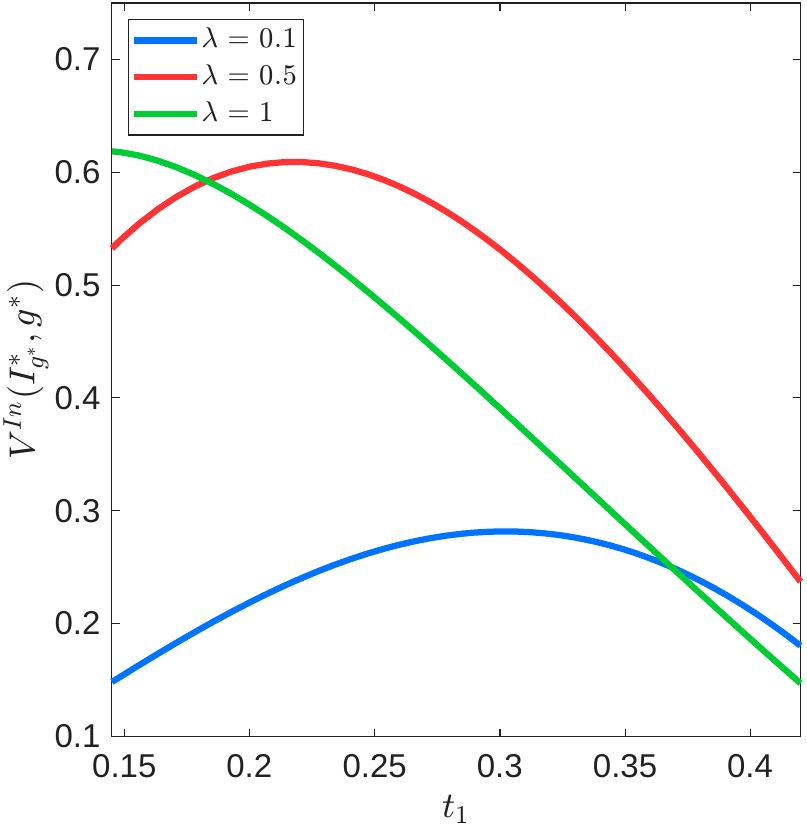} 
\caption{Insurer’s Expected Profit Under SE.}
\label{F2-b}
\end{subfigure}
\caption{The case where $X$ follows a truncated exponential distribution.}
\label{Fig:bowley2}
\end{figure}

\medskip

Finally, assume that the random loss $X$ follows a Kumaraswamy distribution with CDF given by
$$
F_X(x) = 1 - \big(1 - (\tfrac{x}{M})^a\big)^b, \ \forall x \in [0, M], \ M=10, \ \text{and $a,b>0$}.
$$

\medskip

Figure \ref{F3-a} plots the CDFs for different combinations of $a$ and $b$. The blue curve lies close to, but slightly below, the identity function, indicating that more losses are concentrated on the right compared to the uniform distribution. The red curve is more concave than the blue one, implying that larger losses are more likely. The green curve represents the most extreme right-skewed case among the three. 

\medskip

Figure \ref{F3-b} shows the corresponding insurer's profit under these three scenarios. When $a = 1.5$ and $b = 1$, the insurer's profit does not vary monotonically with $t_1$, and it reaches its maximum around $t_1 \approx 0.32$. When $a = 1.5$ and $b = 0.5$, meaning the loss distribution becomes riskier in the sense of first-order stochastic dominance, the insurer's profit decreases for all $t_1$. This suggests that when the underlying risk increases, the insurer's profit declines regardless of the policyholder's risk attitude. Moreover, the maximum shifts to the right, reaching about $t_1 \approx 0.37$. When the loss becomes even riskier with $a = 2$ and $b = 0.3$, a similar decreasing pattern is observed, with the maximum shifting further to the right to approximately  $t_1 \approx 0.4$.

\begin{figure}[H]
\centering
\begin{subfigure}{0.48\textwidth}
\centering
\includegraphics[height=8cm]{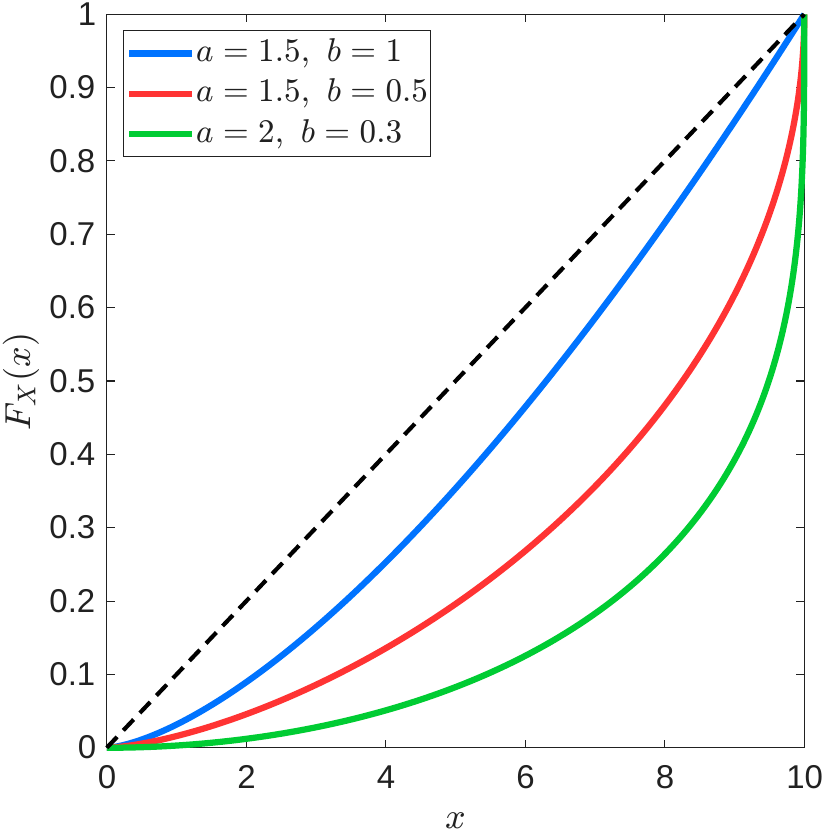}
\caption{Cumulative Distribution Function.}
\label{F3-a}
\end{subfigure}\hfill
\begin{subfigure}{0.48\textwidth}
\includegraphics[height=8cm]{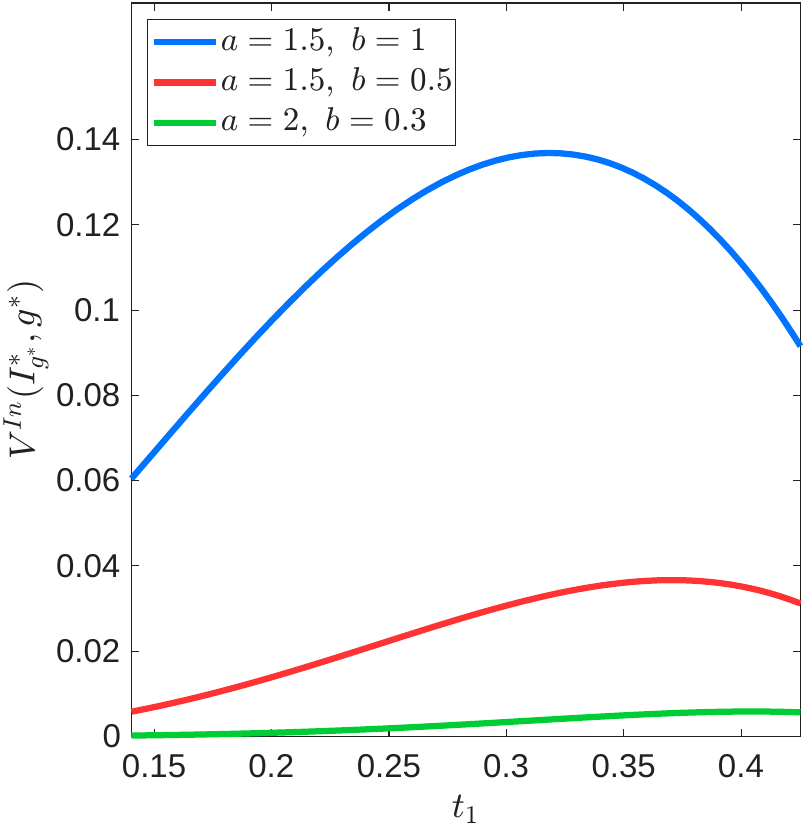}
\caption{Insurer's Expected Profit Under SE.}
\label{F3-b}
\end{subfigure}
\caption{The case where $X$ follows a Kumaraswamy distribution.}
\label{Fig:bowley3}
\end{figure}

%============================================
%============================================
%============================================
\bigskip
\section{Conclusion}
\label{SecConDis}

In this paper, we study Stackelberg Equilibria (Bowley optima) in a monopolistic centralized sequential-move insurance market. We consider a risk-neutral, profit-maximizing insurer who sets premia using a distortion premium principle, and a single policyholder who seeks to minimize a distortion risk measure. 

\medskip

We characterize Stackelberg equilibria explicitly, and we show that, in equilibrium, the optimal indemnity function exhibits a layer-type structure, providing full insurance over any loss layer on which the policyholder is more pessimistic than  the insurer's pricing functional about tail losses. Equilibrium contracts provide no coverage over loss layers on which  the policyholder is less pessimistic than the insurer's pricing functional about tail losses. When the policyholder and the insurer's pricing functional are equally pessimistic about tail losses, the marginal indemnity may take an arbitrary shape, within the global feasibility constraints.
Additionally, in equilibrium, the optimal pricing distortion function is determined by the policyholder's degree of (weak) risk aversion, that is, whether the policyholder's distortion function is above or below the identity function, such that prices never exceed the policyholder's marginal willingness to insure tail losses. 

\medskip

When the policyholder is either weakly risk averse (with a distortion function that lies above the identity function) or strongly risk averse (with a concave distortion function), we show that the optimal contract provides full insurance, and the optimal pricing distortion function coincides with that of the policyholder. For policyholders who evaluate risk using a VaR risk measure at a confidence level $\alpha \in (0,1)$, the optimal coverage includes an upper limit, by providing full insurance for losses below the $\alpha$-quantile while leaving the upper tail uninsured. For a policyholder who displays strong sensitivity to extreme losses, captured by an inverse-S-shaped distortion function, the optimal indemnity function takes the form of a deductible contract, whereby extreme losses are fully transferred to the insurer. Furthermore, we show that in a Stackelberg equilibrium, both the insurance coverage and the insurer's expected profit increase with the policyholder's degree of strong or weak risk aversion. A more risk-averse policyholder will receive greater insurance coverage under an equilibrium optimal contract, and is more valuable to the insurer. 

\medskip

Moreover, we examine the Pareto efficiency of equilibrium contracts. Echoing recent work in the literature, we show that equilibrium contracts are Pareto optimal, but they do not induce a welfare gain to the policyholder, which is unsurprising in a monopoly. Conversely, any Pareto-optimal contract that leaves no welfare gain to the policyholder can be obtained as a Stackelberg equilibrium contract.

%=====================================================================
%=====================================================================
%=====================================================================
\bigskip
\medskip
\setlength{\parskip}{0.5ex}
\hypertarget{LinkToAppendix}{\ }
\appendix

\vspace{-0.4cm}

\section{Proofs of Main Results}\label{appendix-proofs}

%=====================================================================
\subsection{Proof of Lemma \ref{le:optimality_qg_iff_Ig}}\label{App:le:optimality_qg_iff_Ig}
For a given pricing distortion function $g$, assume that $q_g$ is optimal for \eqref{DisMin}, and we aim to show that $I_g(x)= x - q_g \left(F_X(x)\right)$ is optimal for \eqref{prob:Pol1}. Consider any feasible solution $I$ to \eqref{prob:Pol1}, and let $q(t)$ be the quantile function of the associated retention $R(X)$, for all $t \in [0,1]$. We saw in \eqref{eq:pol_quantile_rep} that
\begin{align*}
\rho^{Pol}(I,g)
&=\rho^{Pol} (R(X)) + \Pi_g (I(X)) \\
&=\int_0^1 q^\prime(t) \, \left[T(1-t) - g(1-t)\right]\,\mathrm{d}t + \int_0^1 \left(F_{X}^{-1}\right)^\prime(t) \, g(1-t)\,\mathrm{d}t.
\end{align*}

\noindent Since $q_g$ is optimal for the problem given in \eqref{DisMin}, then the following inequality holds:
\begin{align*}
\rho^{Pol}(I,g)
&\geq \int_0^1 q_g^\prime(t) \, \left[T(1-t) - g(1-t)\right]\,\mathrm{d}t + \int_0^1 \left(F_{X}^{-1}\right)^\prime(t) \, g(1-t)\,\mathrm{d}t,
\end{align*}

\noindent where, by Remark \ref{re:quantilecoadd},  
$$q_g(t) = F_X^{-1}(t) - F^{-1}_{I_g(X)}(t), \ \forall t \in [0,1].$$ 

\noindent Hence, 
\begin{align*}
\rho^{Pol}(I,g)
&\geq \int_0^1 \big[ F_X^{-1}(t)-F^{-1}_{I_g(X)}(t) \big]^\prime 
\left[T(1-t) - g(1-t)\right]\,\mathrm{d}t + \int_0^1 \left(F_{X}^{-1}\right)^\prime(t) \, g(1-t)\,\mathrm{d}t \\
&=\int_0^1\left(F_X^{-1}\right)^\prime(t)\,T(1-t)\,\mathrm{d}t-\int_0^1\left(F^{-1}_{I_g(X)}\right)^\prime(t)\,T(1-t) \, \mathrm{d}t + \pi_g\left(I_g(X)\right)\\
&=\int \left(X-I_g(X)+\Pi_g\left(I_g(X)\right)\right)
\mathrm{d}T\circ\mathbb{P}
=\rho^{Pol}\left(I_g,g\right),
\end{align*}

\noindent which implies that the indemnity function $I_g(x) = x - q_g \big( F_X (x) \big)$ is optimal for the policyholder's problem \eqref{prob:Pol1}.

\medskip

Conversely, we assume that the indemnity function $I_g$, given by $I_g(x) = x - q_g \big( F_X (x) \big)$ is optimal for the policyholder's problem \eqref{prob:Pol1}, and we aim to show that the quantile $q_g$ is optimal for \eqref{DisMin}. Consider any feasible solution $q$ to \eqref{DisMin} and let $I(x) = x - q(F_X(x))$. We have that:
\begin{align*}
\rho^{Pol}\left(q,g\right)
&=\int_0^1 q^\prime(t) \, \left[T(1-t) - g(1-t)\right]\,\mathrm{d}t + \int_0^1 \left(F_{X}^{-1}\right)^\prime(t) \, g(1-t)\,\mathrm{d}t \\
&=\int\left(X- I(X) + \Pi_g\left(I(X)\right)\right) \,
\mathrm{d}T\circ\mathbb{P}. 
\end{align*}

\noindent Since $I_g$ is optimal for \eqref{prob:Pol1}, then the following inequality holds:
\begin{align*}
\rho^{Pol}\left(q,g\right)
&\geq \int\left(X-I_g(X)+\Pi_g\left(I_g(X)\right)\right) \,
\mathrm{d}T\circ\mathbb{P}\\
&= \int_0^1 q_g^\prime(t) \, \left[T(1-t) - g(1-t)\right]\,\mathrm{d}t + \int_0^1 \left(F_{X}^{-1}\right)^\prime(t) \, g(1-t)\,\mathrm{d}t
=\rho^{Pol}\left( q_g,g\right),
\end{align*}

\noindent which implies the optimality of $q_g$ for \eqref{DisMin}. \qed

%=====================================================================
\bigskip
\subsection{Proof of Lemma \ref{le:qg_characterization}} \label{App:le:qg_characterization}
First, note that
\begin{align*}
\rho^{Pol}(q,g)
&=\Pi_g(q)+\int_0^1q^\prime(t)\,T(1-t)\mathrm{d}t\\
&=\int_0^1\left(\left(F_X^{-1}\right)^\prime(t)-q^\prime(t)\right)\,g(1-t)\,\mathrm{d}t+\int_0^1q^\prime(t)\,T(1-t)\mathrm{d}t\\
&=\int_0^1\left(F_X^{-1}\right)^\prime(t)\,g(1-t)\,\mathrm{d}t+\int_0^1q^\prime(t)\,\left(T(1-t)-g(1-t)\right)\mathrm{d}t.
\end{align*}

\noindent The first term in the final line is independent of $q$, so minimizing the risk measure $\rho^{Pol}(q,g)$ over $q\in\mathcal{Q}_L$  is equivalent to solving:
$$
\min\limits_{q\in \mathcal{Q}_L}\int_0^1 q^\prime(t)\left(T(1-t)-g(1-t)\right)\,\mathrm{d}t.
$$

\noindent By the Marginal Indemnity Function Approach given in \cite{ASSA201570}, the optimal quantile function $q_g$ must satisfy:
\begin{equation}\label{quantile}
\left(q_g\right)^\prime(t)=
\left\{
\begin{array}[c]{ll}%
0, & g(1-t)<T(1-t),\vspace{0.2cm}\\
\in\left[0,\left(F_X^{-1}\right)^\prime(t)\right], & g(1-t)=T(1-t),\vspace{0.2cm}\\
\left(F_X^{-1}\right)^\prime(t), & g(1-t)>T(1-t).
\end{array}
\right.  
\end{equation} \qed

%=====================================================================
\bigskip
\subsection{Proof of Theorem \ref{OptDis}}
\label{AppProofThOptDis}
First, note that the insurer's profit can be expressed as:
\begin{equation}
\begin{split}
\label{INe}
V^{In}\left( q_{{g}},g\right)
&=\int_0^1\left[\left(F_X^{-1}\right)^\prime(t)-q_{{g}}^\prime(t)\right]\,g(1-t)\,\mathrm{d}t-\int_{0}^1 \left[F_X^{-1}(t)-q_{{g}}(t)\right]\,\mathrm{d}t \\
&=\int_0^1\left[\left(F_X^{-1}\right)^\prime(t)-q_{{g}}^\prime(t)\right]\,g(1-t)\,\mathrm{d}t-\int_{0}^1 \left[\left(F_X^{-1}\right)^\prime(t)-q_{{g}}^\prime(t)\right](1-t)\,\mathrm{d}t\\
&= \int_0^1\left[\left(F_X^{-1}\right)^\prime(t)-q_{{g}}^\prime(t)\right] h(t) \ dt,
\end{split}
\end{equation}

\noindent where $h(t):=t-1+g(1-t)$, for all $t \in [0,1]$. Consider now the following three sets:
$$
\cA_1 =\left\{t \in [0,1]:  h(t)< t-\tilde T(t)\right\},
\ 
\cA_2 =\left\{t\in [0,1]: h(t)=t- \tilde T(t)\right\}, 
 \ \hbox{and} \ 
\cA_3 = \left\{t\in [0,1]: h(t)> t-\tilde T(t)\right\}.
$$

\noindent where $\tilde T$ is the conjugate of the distortion function $T$. Then it follows from Lemma \ref{le:qg_characterization} that
\begin{equation}\label{h}
\big( q_h \big)^\prime(t)=
\left\{
\begin{array}[c]{ll}
0, &t\in\cA_1,\\
\phi_{h}(t)\in\left[0,\left(F_X^{-1}\right)^\prime(t)\right], & t\in\cA_2,\\
\left(F_X^{-1}\right)^\prime(t), &t\in\cA_3.
\end{array}
\right.  
\end{equation}

\noindent The insurer's profit in \eqref{INe} can then be rewritten as follows:
\begin{align}
\label{InFiner}
V^{In}\left( q_{{h}},h\right)
&= \int_{ 0}^1\left(F_X^{-1}\right)^\prime(t)\,h(t)\,\One_{\cA_1}(t)\,\mathrm{d}t+\int_{0}^1\left[\left(F_X^{-1}\right)^\prime(t)-\phi_{h}(t)\right]\,h(t)\,\One_{\cA_2}(t)\,\mathrm{d}t.
\end{align}

\medskip

Note that on $\cA_2$, the optimal retention function allows for some flexibility as long as $q_h\in\mathcal{Q}_L$. The arbitrary choice of $\phi_h(t)$ does not affect the policyholder's risk exposure level, but it does impact the insurer's profit. To maximize the insurer's profit, we must take a further step to determine the value of $\phi_h$ when $t \in \cA_2$. This is achieved by analyzing the profit over a finer partition $\left\{\cB_i\right\}_{i=1}^3$ such that:
$$
\cB_1=\{t \in [0,1]: \ \tilde T(t)<t\}, \ \ 
\cB_2=\{t \in [0,1]: \ \tilde T(t)=t\}, \ \  
\hbox{and} \ \ 
\cB_3=\{t \in [0,1]: \ \tilde T(t)>t\}.
$$

\noindent The insurer's profit in \eqref{InFiner} becomes:
\begin{align}
\label{InFiner2}
V^{In}\left( q_{{h}},h\right) &=\int_0^1\left(F_X^{-1}\right)^\prime(t)\,h(t)\,\One_{\mathcal{A}_1}(t)\,\mathrm{d}t
+\sum_{i=1}^3 \int_0^1\left[\left(F_X^{-1}\right)^\prime(t)-\phi_{h}(t)\right]h(t)\,\One_{\mathcal{A}_2}(t) \cdot  \One_{\cB_i}(t ) \ \mathrm{d}t \nonumber  \\
&= \int_0^1\left(F_X^{-1}\right)^\prime(t)\,h(t)\,\One_{\mathcal{A}_1}(t)\,\mathrm{d}t
+\sum_{i=1}^3 \int_0^1\left[\left(F_X^{-1}\right)^\prime(t)-\phi_{h}(t)\right]h(t)\,\One_{\mathcal{A}_2\cap \cB_i}(t) \ \mathrm{d}t.
\end{align}

\noindent We consider the following three cases:
\begin{enumerate}
\item  On $\cA_2 \cap \cB_1$, $h(t)=t-\tilde T(t)>0$. Thus, the insurer's profit in \eqref{InFiner2} is decreasing in  $\phi_h(t)$, and it is optimal to set $\phi_{h}(t)=0$.

\smallskip

\item  On $\cA_2 \cap \cB_2$, $h(t)=t-\tilde T(t) = 0$. The profit contribution is always zero regardless of the value of $\phi_{h} (t)$. Thus, $\phi_{h} (t)$ can take any value in $\left[0,\left(F_X^{-1}\right)^\prime(t)\right]$ without affecting the insurer's profit.

\smallskip

\item On $\cA_2 \cap \cB_3$, $h(t)=t-\tilde T(t)<0$. The insurer's profit is increasing in $\phi_h(t)$. Thus, it is optimal to set $\phi_{h}(t)=\left(F_X^{-1}\right)^\prime(t)$. 
\end{enumerate}

\medskip

Hence, for a given $h$, the marginal quantile of the optimal retention can be written as:
\begin{equation} 
\label{eq:h}
\left(q^*_{h}\right)^\prime(t)=
\left\{
\begin{array}[c]{ll}%
0, &t\in  \cA_2\cap\cB_1,\vspace{0.2cm}\\
\phi_{h}(t)\in\left[0,\left(F_X^{-1}\right)^\prime(t)\right], & t\in\cA_2\cap\cB_2,\vspace{0.2cm}\\
\left(F_X^{-1}\right)^\prime(t), &t\in \cA_2\cap\cB_3.\vspace{0.2cm}
\end{array}
\right.  
\end{equation}

\noindent It then follows from \eqref{h} and \eqref{eq:h} that:
\begin{equation}\label{h2}
\left(q^*_{h}\right)^\prime(t)=
\left\{
\begin{array}[c]{ll}
0, &t\in\cA_1 \cup \left(\cA_2 \cap \cB_1\right),\vspace{0.2cm}\\
\phi_{h}(t)\in\left[0,\left(F_X^{-1}\right)^\prime(t)\right], & t\in\cA_2 \cap\cB_2,\vspace{0.2cm}\\
\left(F_X^{-1}\right)^\prime(t), &t\in\cA_3\cup\left(\cA_2 \cap\cB_3\right).\vspace{0.2cm}
\end{array}
\right.  
\end{equation}

\noindent Using \eqref{h2} and the fact that $h(t) =0$ on $\cA_2 \cap \cB_2$, the insurer's profit in \eqref{InFiner2} reduces to:
\begin{align*}
V^{In}\left( q^*_{h},h\right)
&= \int_0^1 (F_X^{-1})'(t)\,h(t)
\bigl[ \One _{\cA_1}(t) + \One_{\mathcal A_2\cap\mathcal B_1}(t)
\bigr]\ \mathrm{d}t \\
&=\int_0^1\left(F_X^{-1}\right)^\prime(t)\, h(t) \ \One_{\mathcal{A}_1\cup \big( \mathcal{A}_2\cap \cB_1 \big)}(t) \ \mathrm{d}t , \ \text{since $\cA_1$ and $\cA_2$ are disjoint.}
\end{align*}

\noindent Moreover,
\begin{align*}
\cA_1 \cup \big( \cA_2 \cap \cB_1 \big) 
&= \big( \cA_1  \cup \cA_2  \big) \  \cap \ \big( \cA_1 \cup \cB_1\big). 
\end{align*}

\noindent Since $\{\cB_i\}_{i=1}^3$ forms a partition of $[0,1]$, we have:
$$\mathcal A_1 \cup\mathcal B_1
=
\mathcal B_1 \cup \left(\mathcal A_1 \cap(\mathcal B_2\cup\mathcal B_3)\right).$$

\noindent Combining the above identities yields:
$$\cA_1\cup(\cA_2\cap\cB_1)
=
\left[(\cA_1 \cup \cA_2) \cap\cB_1\right]
\cup
\left[\cA_1 \cap ( \cB_2 \cup \cB_3)\right],$$

\noindent which is a union of disjoint sets. Therefore, 
$$V^{In}(q_h^*,h)
=
\int_0^1 (F_X^{-1})'(t)\,h(t)
\left[
\mathbf 1_{(\cA_1 \cup \cA_2) \cap \cB_1}(t)
+
\mathbf 1_{\cA_1 \cap (\cB_2 \cup \cB_3)}(t)
\right] \ \mathrm{d}t.$$

\noindent The insurer's optimization problem given in \eqref{IN3} thus reduces to the following:
\begin{align}\label{INh}
\max\limits_{h} \ \left\{\int_0^1\left(F_X^{-1}\right)^\prime(t)\, h(t)\left[ \One_{\{h(t)\leq t-\tilde T(t)\}}\One_{\{\tilde T(t)<t\}}+\One_{\{h(t)< t-\tilde T(t)\}}\One_{\{\tilde T(t)\geq t\}}\right]\mathrm{d}t\right\}.
\end{align}

\medskip

Since the integrand depends on $h$ only through the pointwise value $h(t)$, the above problem \eqref{INh} can be solved by maximizing the integrand pointwise. Fix $t_0\in(0,1)$ and consider the auxiliary maximization problem:
$$\max_{y\geq 0} \ m(y;t_0),$$

\noindent where the auxiliary function $m(y; t_0)$ is given by
$$
m(y;t_0) 
:= y \left[ \One_{\{y\leq t_0-\tilde T(t_0)\}}\One_{\{ \tilde T(t_0)<t_0\}} + \One_{\{y<t_0-\tilde T(t_0)\}} \One_{\{ \tilde T(t_0) \geq t_0\}}\right].
$$

\medskip

\noindent For a fixed $t_0$, the function $m(y; t_0)$ is piecewise linear in $y$. Let $y_{t_0}\in\underset{y}\argmax\, m(y;t_0)$. The maximum value of $m(y;t_0)$ is given by
\begin{equation*}
m(y_{t_0};t_0)=
\left\{
\begin{array}[c]{ll}%
t_0-\tilde T(t_0), &\tilde T(t_0)<t_0,\vspace{0.2cm}\\
0,& \tilde T(t_0)\geq t_0,
\end{array}
\right.  
\end{equation*}

\noindent and any maximizer $y_{t_0}$ satisfies:
\begin{equation} 
\label{dual_y}
\begin{cases}
y_{t_0} =  t_0-\tilde T(t_0), \ & \text{ if $\tilde T(t_0)<t_0$}, \\
y_{t_0} \geq  t_0-\tilde T(t_0), \ & \text{ if $\tilde T(t_0) \geq t_0$}. 
\end{cases}
\end{equation}

\noindent for any $t_0\in (0,1)$. We now proceed to the characterization of optimal solutions to \eqref{INh} through the following Lemma.  

\medskip

\begin{lemma}
A function $h^*$ is optimal to \eqref{INh} if and only if for every $t_0 \in (0,1)$, $h^*(t_0)=y_{t_0}$, where $y_{t_0}$ satisfies \eqref{dual_y}.
\end{lemma}

\begin{proof}
Assume that $h^*$ is a solution to \eqref{INh}. We show that $h^* (t_0)= y_{t_0}$, where $y_{t_0}$ satisfies \eqref{dual_y}. We consider the following two cases. 
\medskip
\begin{enumerate}
\item If $\tilde T(t_0)<t_0,$ we aim to show that $h^*(t_0)=t_0-\tilde T(t_0)$ for any $t_0\in(0,1)$. Suppose for the sake of contradiction that
$$h^*(t_0)\neq t_0-\tilde T(t_0).$$

\noindent Since $\tilde T(t_0)<t_0$ and $\tilde T$ is continuous (being differentiable), we can find an arbitrary small  $\epsilon>0$ such that 
$$
\tilde T(t)<t, \ \ \text{when $t\in(t_0-\epsilon,t_0+\epsilon)$}.
$$

\noindent If $h^*(t_0)<t_0-\tilde T(t_0)$, we have $m(h^*(t_0);t_0)=h^*(t_0)$. Then there exists a function $\tilde h$ such that:
$$
\tilde h(t) =
\begin{cases}
y_{t}>h^*(t), & t\in(t_0-\epsilon,t_0+\epsilon), \\
h^*(t), & t\notin(t_0-\epsilon,t_0+\epsilon),
\end{cases}
$$

\noindent where $y_t$ satisfies \eqref{dual_y}, which implies the following:
\begin{align*}
V^{In}\left( q^*_{\tilde h},\tilde h\right)-V^{In}\left( q^*_{h^*},h^*\right)   
&=\int_{t_0-\epsilon}^{t_0+\epsilon}\left(F_X^{-1}\right)^\prime(t)\, \left[m(y_{t};t)- m\left(h^*(t);t\right) \right]\mathrm{d}t\\
&= \int_{t_0-\epsilon}^{t_0+\epsilon}\left(F_X^{-1}\right)^\prime(t)\, \left[y_t- h^*(t)\right]\mathrm{d}t
>0,
\end{align*}

\noindent since $y_t>h^*(t)$ when $t\in(t_0-\epsilon,t_0+\epsilon)$. Hence, this contradicts the fact that $h^*$ is optimal for problem \eqref{INh}.

\medskip

\noindent If, in contrast, $h^*(t_0)>{t_0}-\tilde T(t_0)$, then $m\left(h^*(t_0);t_0\right)=0$. With the same $\epsilon$, there exists a function $\tilde h$ such that:
$$
\tilde h(t)= \begin{cases}
y_{t}<h^*(t) , &   t\in(t_0-\epsilon,t_0+\epsilon), \\
h^*(t), & t\notin(t_0-\epsilon,t_0+\epsilon).
\end{cases}
$$ 

\noindent Hence, 
\begin{align*}
V^{In}\left( q^*_{\tilde h},\tilde h\right)-V^{In}\left( q^*_{h^*},h^*\right)   
&=\int_{t_0-\epsilon}^{t_0+\epsilon}\left(F_X^{-1}\right)^\prime(t)\, \left[m(y_{t};t)- m\left(h^*(t);t\right)\right]\mathrm{d}t\\
&= \int_{t_0-\epsilon}^{t_0+\epsilon}\left(F_X^{-1}\right)^\prime(t)\, \left(y_t- 0\right)\mathrm{d}t
> 0,
\end{align*}

\noindent which also contradicts the optimality of $h^*$ for \eqref{INh}. Therefore,  $h^*(t_0)=t_0-\tilde T(t_0)$, for any $t_0\in(0,1)$, when $\tilde T(t_0)<t_0$.

\bigskip

\item If $\tilde T(t_0)\geq t_0,$ we aim to show that $ h^*(t_0)\geq t_0-\tilde T(t_0)$, for any $t_0 \in (0,1)$. Suppose for the sake of contradiction that
$$h^*(t_0)<t_0-\tilde T(t_0).$$ 

Then
$$m\left(h^*(t_0);t_0\right) = h^*(t_0)< t_0-\tilde T(t_0) \leq 0.$$

\medskip

There exists an arbitrary small  $\epsilon > 0$ such that $\tilde T(t)\geq t$ on $(t_0-\epsilon,t_0]$ or on $[t_0,t_0+\epsilon)$. First, assume that $\tilde T(t)\geq t$ on  $(t_0-\epsilon,t_0]$. Then there exists a function $\tilde h$  such that:
$$
\tilde h(t)=
\begin{cases}
y_{t}>h^*(t), & t\in(t_0-\epsilon,t_0], \\
h^*(t), & t\notin(t_0-\epsilon,t_0]. 
\end{cases}
$$

\noindent Therefore,
\begin{align*}
V^{In}\left( q^*_{\tilde h},\tilde h\right)-V^{In}\left( q^*_{h^*},h^*\right)   
&=\int_{t_0-\epsilon}^{t_0}\left(F_X^{-1}\right)^\prime(t)\, \left[m(y_{t};t)- m\left(h^*(t);t\right)\right]\mathrm{d}t\\
&= \int_{t_0-\epsilon}^{t_0}\left(F_X^{-1}\right)^\prime(t)\, \left[0- m\left(h^*(t);t\right)\right]\mathrm{d}t
>0,
\end{align*}

\noindent since $m\left(h^*(t);t\right)<0$ when $t\in(t_0-\epsilon,t_0]$, which leads to a contradiction. Hence, 
$$h^*(t_0)\geq t_0-\tilde T(t_0).$$

\noindent Moreover, we can derive a similar result if $\tilde T(t)\geq t$ on  $[t_0,t_0+\epsilon)$. In sum, for any $t_0\in[0,1]$, we obtain:
\begin{equation}
\label{dual_h}
\begin{cases}
h^*(t_0)=t_0-\tilde T(t_0), & \text{ if} \ \tilde T(t_0)<t_0,\vspace{0.2cm}\\
h^*(t_0) \geq t_0-\tilde T(t_0), &  \text{ if} \ \tilde T(t_0)\geq t_0.\vspace{0.2cm}
\end{cases}
\end{equation}
\end{enumerate}

\bigskip

Conversely, assume that $h^*$ satisfies \eqref{dual_h}. We show that $h^*$ is optimal for \eqref{INh}. We first note that the insurer's expected profit under $(q^*_{h^*}, h^*)$ is given by
$$
V^{In}\left( q^*_{h^*},h^*\right) = \int_0^1\left(F_X^{-1}\right)^\prime(t)\,(t-\tilde T(t))\One_{\{\tilde T(t)<t\}} \, \mathrm{d}t. 
$$

\noindent For any other feasible solution $h$, we compare the insurer's expected profit under $h^*$ and $h$. Specifically, from the structure of \eqref{INh}, we have:
\begin{align*}
&V^{In}\left( q^*_{h^*},h^*\right)-V^{In}\left( q^*_{h},h\right)  \\
& \quad 
=\int_0^1\left(F_X^{-1}\right)^\prime(t)\,(t-\tilde T(t))\One_{\{\tilde T(t)<t\}}\mathrm{d}t\\
& \qquad  -\int_0^1\left(F_X^{-1}\right)^\prime(t)\, h(t) \bigg[\One_{\{h(t)\leq t-\tilde T(t)\}}\One_{\{\tilde T(t)<t\}}(t)+\One_{\{h(t)<t-\tilde T(t)\}}\One_{\{\tilde T(t)= t\}}(t) +\One_{\{h(t)<t-\tilde T(t)\}}\One_{\{\tilde T(t)> t\}}(t)\bigg]\mathrm{d}t. 
\end{align*}

\noindent That is, 
\begin{align}\label{InOpt}
V^{In}\left( q^*_{h^*},h^*\right)-V^{In}\left( q^*_{h},h\right) 
&=\int_0^1\left(F_X^{-1}\right)^\prime(t)\,\left[t-\tilde T(t)-h(t)\One_{\{h(t)\leq t-\tilde T(t))\}}\right]\One_{\{\tilde T(t)<t\}}\mathrm{d}t\nonumber\\
&\qquad-\int_0^1\left(F_X^{-1}\right)^\prime(t)\, h(t) \left[\One_{\{h(t)<t-\tilde T(t)\}}\One_{\{\tilde T(t)= t\}}+\One_{\{h(t)<t-\tilde T(t)\}}\One_{\{\tilde T(t)> t\}}\right]\mathrm{d}t.
\end{align}

\noindent Looking at the first term of \eqref{InOpt}, we know that
\begin{align*}
t-\tilde T(t)-h(t)\One_{\{h(t)\leq t-\tilde T(t)\}}&=\left\{
\begin{array}[c]{ll}
t-\tilde T(t),& h(t)>t-\tilde T(t)\vspace{0.2cm}\\
t-\tilde T(t)-h(t),&h(t)\leq t-\tilde T(t),
\end{array}
\right.  
\end{align*}

\noindent which is always nonnegative. For the second term,
$$
-h(t) \One_{\{h(t)<t-\tilde T(t)\}} \
\begin{cases}
=h(t) \One_{\{h(t)<0\}}\geq 0, & \text{if }\tilde T(t)=t,\\
\geq 0, & \text{if } \tilde T(t)>t. 
\end{cases}
$$

\noindent Thus, both terms of \eqref{InOpt} are nonnegative, for all $t\in[0,1]$, implying that the integrand is pointwise nonnegative. Therefore, the difference in profit is nonnegative, i.e.
$$V^{In}\left( q^*_{h^*},h^*\right)\geq V^{In}\left( q^*_{h},h\right).$$

\noindent Hence, $h^*$ maximizes the insurer's profit functional, thereby establishing sufficiency.
\end{proof}

\bigskip

Moreover, when $h^*$ satisfies \eqref{dual_h}, we can characterize the structure of the critical sets $\cA_1 \cup (\cA_2 \cap \cB_1)$, $\cA_2 \cap \cB_2$, and $\cA_3 \cup ( \cA_2 \cap \cB_3)$ that define $q^*_{h^*}$ in \eqref{h2}. First, if $h^*$ satisfies \eqref{dual_h}, we can clearly see that $\cA_1 = \varnothing$. Hence, 
$$\cA_1 \cup \left(\cA_2 \cap\cB_1\right) =\{t \in [0,1]:  t-\tilde T(t)>0\}.$$

\noindent Moreover,
$$\cA_2\cap\cB_2=\{t \in [0,1]:   h^*(t)=t-\tilde T(t)=0\},$$ 
and 
\begin{align*}  
\cA_3 \cup \left(\cA_2 \cap\cB_3\right)
&= \{t \in [0,1]:  h^*(t)>t-\tilde T(t)\} \  \cup \ \{t\in [0,1]: h^*(t)=t-\tilde T(t)<0\}. 
\end{align*}

\noindent Additionally, since $\{ \cB_i\}_{i=1}^3$ forms a partition over $(0,1)$, it follows that
\begin{align*}
\cA_3 
&= (\cA_3 \cap \cB_1) \ \cup \  (\cA_3 \cap \cB_2) \ \cup \  (\cA_3 \cap \cB_3) \\
&= \{t \in [0,1]: h^*(t) > t - \tilde T(t), t > \tilde T(t)\} \\
&\qquad\qquad \cup  \ \{ t\in [0,1]: h^*(t) > t- \tilde T(t) =0\} \\ 
&\qquad\qquad\qquad \cup \{t\in [0,1] : h(t)^* > t - \tilde T(t), \tilde T(t)> t\}. 
\end{align*}

\noindent The first intersection must be empty since $h^*$ satisfies \eqref{dual_h}:
$$
\cA_3 \cap \cB_1 = \varnothing. 
$$

\noindent Hence, the union of $\cA_3$ with $\cA_2 \cap \cB_3$ reduces to the following:
$$\cA_3 \cup \left(\cA_2 \cap\cB_3\right)
=
\{t\in [0,1]: h^*(t)>t-\tilde T(t)=0\} \ \cup \ \{t\in [0,1]:  \tilde T(t)>t\}. $$

\noindent Thus, the function $q^*_{h^*}$ given in \eqref{h2} can be written as follows:
\begin{equation*}
\left(q^*_{h^*}\right)^\prime(t)=
\left\{
\begin{array}[c]{ll}%
0, &\tilde T(t)<t,\vspace{0.2cm}\\
\phi_{h}(t)\in\left[0,\left(F_X^{-1}\right)^\prime(t)\right], &  t\in\{z; \  h^*(z)=z-\tilde T(z)=0\}\\
\left(F_X^{-1}\right)^\prime(t),  &t\in\{z; \  h^*(z)>z-\tilde T(z)=0\} \cup \{ z; \ \tilde T(z)> z \}.\
\end{array}
\right. 
\end{equation*}

\noindent Since $\tilde{g}^*(t)=t-h^*(t)$, the optimal $h^*$  leads to the following optimal pricing distortion function $\tilde g^*$:
\begin{equation}\label{dual_g}
\tilde g^*(t)=
\left\{
\begin{array}[c]{ll}%
\tilde T(t), &\tilde T(t)<t,\vspace{0.2cm}\\
\in\left[0,\tilde T(t)\right], & \tilde T(t)\geq t.
\end{array}
\right.  
\end{equation}

\noindent Consequently, 
\begin{equation}
\tilde g^*(t)=
\left\{
\begin{array}[c]{ll}
\tilde T(t), &\tilde T(t)<t,\vspace{0.3cm}\\
\in\left[\,\sup \left\{z < t ; \ \tilde g^*(z) \right\}, \tilde T(t)\right], & \tilde T(t)\geq t.
\end{array}
\right.  
\end{equation}

\medskip

\noindent We then obtain the optimal retention quantile $q^*$ as a function of the distortion premium function $g^*$, as follows:
\begin{equation*}
\left(q^*_{{g}^*}\right)^\prime(t)=
\left\{
\begin{array}[c]{ll}
0, &\tilde T(t)<t,\vspace{0.2cm}\\
\phi_{ g^*}(t)\in\left[0,\left(F_X^{-1}\right)^\prime(t)\right], & t\in\{z;\ \tilde{g}^*(z)= \tilde T(z)=z\},\vspace{0.2cm}\\
\left(F_X^{-1}\right)^\prime(t), &t\in\{z; \  \tilde{g}^*(z)<\tilde T(z)=z\}\cup\{z; \ \tilde T(z)>z\} .\vspace{0.2cm}
\end{array}
\right.  
\end{equation*} \qed

\bigskip
%=====================================================================
\subsection{Proof of Theorem \ref{CompStatWRA}}
\label{AppProofThCompStatWRA}

Consider two policyholders whose respective distortion functions $T_1$ and $T_2$ satisfy $T_1(t) \leq T_2(t)$ for all $t \in [0,1]$. Then
$$T_1(\mathbb{P}[X > y]) \leq T_2(\mathbb{P}[X > y]), \quad \forall\, y \in [0,M].$$

\medskip

Let $(I_{g^*_1}^*, g_1^*)$ and $(I_{g^*_2}^*, g_2^*)$ denote the corresponding Stackelberg equilibria, where $\kappa_1^*$ and $\kappa_2^*$ denote the marginal indemnity functions of the two policyholders, respectively, and satisfy Corollary~\ref{ThStep2}. Hence, we obtain
$$\kappa^*_1 (y) \leq \kappa^*_2(y), \ \text{for almost every $y \in [0,M]$}.$$

\noindent The above inequality gives
$$
I_{g^*_1}^*(x) \le I_{g^*_2}^*(x), \ \forall x \in [0,M].
$$

\medskip

\noindent The insurer's expected profits for the two policyholders are given by:
$$
V^{In}(I_{g^*_1}^*, g_1^*) 
= \int_0^1 \big(F_X^{-1}\big)^\prime(t)\, [t - \tilde{T}_1(t)]\, \One_{\{\tilde{T}_1(t) < t\}}\, \mathrm{d}t
\ \ \hbox{ and } \ \ 
V^{In}(I_{g^*_2}^*, g_2^*) 
= \int_0^1 \big(F_X^{-1}\big)^\prime(t)\, [t - \tilde{T}_2(t)]\, \One_{\{\tilde{T}_2(t) < t\}}\, \mathrm{d}t,
$$

\noindent where $\tilde T_1$ and $\tilde T_2$ denote the conjugates of $T_1$ and $T_2$, respectively. Since $T_1(t) \leq T_2(t)$ for all $t \in [0,1]$, it follows that $\tilde{T}_1 (t) \geq \tilde{T}_2 (t)$, for all $t \in [0,1]$. Hence,
$$t - \tilde{T}_1(t) \leq t - \tilde{T}_2(t), \ \forall t \in [0,1],$$ 
and,
$$\{\tilde{T}_1(t) < t\} \subseteq \{\tilde{T}_2(t) < t\}.$$

\noindent Therefore, 
\begin{align*}
V^{In}(I_{g^*_1}^*, g_1^*) 
&= \int_0^1 \big(F_X^{-1}\big)^\prime(t)\, [t - \tilde{T}_1(t)]\, \One_{\{\tilde{T}_1(t) < t\}}\, \mathrm{d}t\\
&\leq \int_0^1 \big(F_X^{-1}\big)^\prime(t)\, [t - \tilde{T}_2(t)]\, \One_{\{\tilde{T}_2(t) < t\}}\, \mathrm{d}t
=V^{In}(I_{g^*_2}^*, g_2^*).
\end{align*}

\noindent Hence, the second policyholder is more profitable for the insurer than the first. \qed

%=====================================================================
\bigskip
\subsection{Proof of Proposition \ref{ChaPo}}
\label{AppProofChaPo}

Let $(I^*, \pi^*) \in \mathcal{I}_L \times \mathbb{R}$ be optimal for the problem given in \eqref{ChaPoPro}, and assume for the sake of contradiction that $(I^*, \pi^*)$ is not Pareto optimal. Then there exists a contract $(I, \pi) \in \mathcal{I}_L \times \mathbb{R}$ such that 
$$ 
\rho^{{Pol}}(I, \pi) \leq \rho^{{Pol}}(I^*, \pi^*) \quad \text{and} \quad V^{{In}}(I, \pi) \geq V^{{In}}(I^*, \pi^*),
$$

\noindent with at least one strict inequality. Consequently,
$$
\rho^{{Pol}}(I, \pi) - V^{{In}}(I, \pi) < \rho^{{Pol}}(I^*, \pi^*) - V^{{In}}(I^*, \pi^*),
$$

\noindent which contradicts the optimality of $(I^*, \pi^*)$ for the problem given in \eqref{ChaPoPro}. Hence, the contract $(I^*, \pi^*)$ is Pareto optimal. 

\medskip

Conversely, suppose that the contract $(I^*, \pi^*)$ is Pareto optimal, and assume for the sake of contradiction that $(I^*, \pi^*)$ is not optimal for \eqref{ChaPoPro}. Then there exists a contract $(\tilde{I}, \tilde{\pi}) \in \mathcal{I}_L \times \mathbb{R}$ such that
$$
\rho^{{Pol}}(\tilde{I}, \tilde{\pi}) - V^{{In}}(\tilde{I}, \tilde{\pi}) 
< \rho^{{Pol}}(I^*, \pi^*) - V^{{In}}(I^*, \pi^*).
$$ 

\medskip

\noindent Let $\hat{\pi} := \tilde{\pi} + \rho^{{Pol}}(I^*, \pi^*) - \rho^{{Pol}}(\tilde{I}, \tilde{\pi})$. Then by translation invariance of $\rho^{Pol}$, we have:
$$
\rho^{{Pol}}(\tilde{I}, \hat{\pi}) = \rho^{{Pol}}(X - \tilde{I}(X)) + \hat{\pi} = \rho^{{Pol}}(I^*, \pi^*).
$$

\noindent Moreover, 
\begin{align*}
V^{{In}}(\tilde{I}, \hat{\pi}) - V^{{In}}(I^*, \pi^*) 
&= \hat{\pi} - \mathbb{E}[\tilde{I}(X)] - V^{{In}}(I^*, \pi^*) \\
&= \left( \tilde{\pi} + \rho^{{Pol}}(I^*, \pi^*) - \rho^{{Pol}}(\tilde{I}, \tilde{\pi}) \right) - \mathbb{E}[\tilde{I}(X)] - V^{{In}}(I^*, \pi^*) \\
&= \rho^{{Pol}}(I^*, \pi^*) - \rho^{{Pol}}(\tilde{I}, \tilde{\pi}) + V^{{In}}(\tilde{I}, \tilde{\pi}) - V^{{In}}(I^*, \pi^*) 
> 0. 
\end{align*}

\noindent Thus, the contract $(\tilde{I}, \hat{\pi})$ satisfies
$$
\rho^{{Pol}}(\tilde{I}, \hat{\pi}) = \rho^{{Pol}}(I^*, \pi^*) \quad \text{and} \quad V^{{In}}(\tilde{I}, \hat{\pi}) > V^{{In}}(I^*, \pi^*),
$$

\noindent which contradicts the Pareto optimality of $(I^*, \pi^*)$. Hence, $(I^*, \pi^*)$ is optimal for \eqref{ChaPoPro}. \qed

%=====================================================================
\bigskip
\subsection{Proof of Proposition \ref{SeIsPo}}
\label{AppProofSeIsPo}
Consider a Stackelberg equilibrium $(I^*_{g^*}, g^*)$ inducing a contract $(I^*_{g^*}, \pi^*_{g^*})$, where $\pi^*_{g^*} := \Pi_{g^*}(I^*_{g^*}(X)) = \int I^*_{g^*} \, d g^* \circ \p$. By Proposition \ref{ChaPo}, to show that $(I^*_{g^*}, \pi^*_{g^*})$ is Pareto optimal, it suffices to show that it is optimal for Problem \eqref{ChaPoPro}. By translation invariance of the distortion risk measure, a contract $(I^*, \pi^*)$ is Pareto optimal if and only if the indemnity function $I^*$ solves
$$\min_{I \in \mathcal{I}_L} \left\{ \rho^{{Pol}}(X - I(X)) + \mathbb{E}[I(X)] \right\}.$$

\noindent Equivalently, using retention functions, the above problem becomes:
$$\min_{R \in \mathcal{I}_L} \left\{ \rho^{{Pol}}(R(X)) + \mathbb{E}[X - R(X)] \right\}.$$

\noindent Using $q(t) := F_{R(X)}^{-1}(t)$ for a.e.\ $t \in [0,1]$, the above problem is equivalent to the following:
\begin{align*}
\min_{q \in \mathcal{Q}_L} \left\{ \int_0^1 q'(t)\left[ T(1 - t) - (1 - t) \right]\, \mathrm{d}t + \mathbb{E}[X] \right\}.
\end{align*}

\medskip

\noindent Since $(I^*_{g^*}, g^*)$ is a Stackelberg equilibrium, the optimal quantile $q^*_{g^*}$ characterized in Theorem \ref{OptDis} solves the above minimization problem. Hence, the induced contract $(I^*_{g^*}, \pi^*_{g^*})$ is Pareto optimal. We now show that $(I^*_{g^*}, \pi^*_{g^*})$ is individually rational but leaves the policyholder indifferent. First, note that the insurer's expected profit under $(I^*_{g^*}, \pi^*_{g^*})$ satisfies
\begin{align*}
V^{In}(I^*_{g^*},\pi^*_{g^*})
&=\int_0^1\left(F_X^{-1}\right)^\prime(t)\left[t-\tilde T(t)\right]\,\One_{\{\tilde T(t)<t\}}\,\mathrm{d}t
\geq 0
=V^{In}(0,0).
\end{align*}

\noindent Moreover, for the policyholder, we have
\begin{align*}
\rho^{Pol}(I^*_{g^*},\pi^*_{g^*})-\rho^{Pol}(0,0)
&=\int_0^1\left(F_X^{-1}\right)^\prime(t)\,g^*(1-t)\,\mathrm{d}t+\int_0^1\left(q^*_{g^*}\right)^\prime(t)\,\left[T(1-t)-g^*(1-t)\right] \, \mathrm{d}t\\
&\qquad\qquad-\int_0^1\left(F_X^{-1}\right)^\prime(t)\,T(1-t)\,\mathrm{d}t\\     &=\int_0^1\left[\left(q^*_{g^*}\right)^\prime(t)-\left(F_X^{-1}\right)^\prime(t)\right]\,\left[T(1-t)-g^*(1-t)\right] \, \mathrm{d}t.
\end{align*}

\medskip

\noindent Consider the following three cases:
\medskip
\begin{enumerate}
\item If $\tilde{T}(t) < t$, then by optimality, $\tilde{g}^*(t) = \tilde{T}(t)$, or equivalently $g^*(1 - t) = T(1 - t)$. Hence, 
$\left[\left(q^*_{g^*}\right)^\prime(t)-\left(F_X^{-1}\right)^\prime(t)\right]\,\left[T(1-t)-g^*(1-t)\right] =0$. 

\medskip

\item If $\tilde{T}(t) = t$ and $\tilde{g}^*(t) < \tilde{T}(t)$, then $\left(q^*_{g^*}\right)^\prime(t) = \left(F_X^{-1}\right)^\prime(t)$. 
If $\tilde{T}(t) = t$ and $\tilde{g}^*(t) = \tilde{T}(t)$, then $g^*(1 - t) = T(1 - t)$. 
In both cases, 
$\left[\left(q^*_{g^*}\right)^\prime(t)-\left(F_X^{-1}\right)^\prime(t)\right]\,\left[T(1-t)-g^*(1-t)\right] =0$. 

\medskip

\item If $\tilde{T}(t) > t$, then $\left(q^*_{g^*}\right)^\prime(t) = \left(F_X^{-1}\right)^\prime(t)$, and $\left[\left(q^*_{g^*}\right)^\prime(t)-\left(F_X^{-1}\right)^\prime(t)\right]\,\left[T(1-t)-g^*(1-t)\right] =0$. 
\end{enumerate}

\noindent Consequently, we finally obtain $\rho^{Pol}(I^*_{g^*},\pi^*_{g^*}) = \rho^{Pol}(0,0)$.\qed

%=====================================================================
\bigskip
\subsection{Proof of Proposition \ref{PoIsSe}}
\label{AppProofPoIsSe}
Consider a Pareto optimal contract $(I^*, \pi^*)$, such that $\rho^{Pol}(I^*,\pi^*) = \rho^{{Pol}}(0,0)$. We show that this contract is induced by a Stackelberg equilibrium. First, note that under $(I^*, \pi^*)$, the following holds by translation invariance and comonotonic additivity of $\rho^{Pol}$:  
\begin{align*}
\rho^{Pol}(I^*,\pi^*) 
&= \rho^{Pol}\left(X-I^*(X)\right) + \pi^*
=\rho^{Pol}\left(X\right) - \rho^{Pol}\left(I^*(X)\right)
+ \pi^*.
\end{align*}

\noindent Since, $\rho^{Pol}(I^*,\pi^*) =\rho^{{Pol}}(0,0)=\rho^{{Pol}}(X)$, we have that $\pi^* = \rho^{{Pol}}(I^*(X))$. Moreover, since $(I^*, \pi^*)$ is Pareto optimal, Proposition \ref{ChaPo} implies that the quantile function $q^*(t):= F_X^{-1}(t) - F_{I^*(X)}^{-1}(t)$ satisfies
\begin{equation*}
(q^*)'(t) =
\begin{cases}
0, & T(1 - t) > 1 - t, \\
\phi(t), & T(1 - t) = 1 - t, \\
\left(F_X^{-1}\right)'(t), & T(1 - t) < 1 - t,
\end{cases}
\end{equation*}

\noindent for some measurable function $\phi$ satisfying $ \phi(t) \in[0,\left(F_X^{-1}\right)'(t)].$ 

\medskip

Consider now the pair $(q^*, g^*_{I^*})$, where the pricing distortion function $g^*_{I^*}$ coincides with the policyholder's distortion function $T$. That is,
$$
g^*_{I^*}(t) := T(t), \ \forall \, t \in [0, 1].
$$

\noindent Thus,
$$
\Pi_{g^*_{I^*}} \left( I^*(X) \right) = \int I^*(X)\, \mathrm{d}g^*_{I^*} \circ \mathbb{P}=\int I^*(X)\, \mathrm{d}T \circ \mathbb{P}= \rho^{{Pol}}(I^*(X))=\pi^*.
$$

\medskip

Since $(q^*,g^*_{I^*})$ satisfies the optimality conditions of Theorem \ref{OptDis}, it follows that $(q^*,g^*_{I^*})$ solves the problem given in \eqref{IN3}. Lemma \ref{LemEq2} then implies that the market mechanism $(I^*,g^*_{I^*})$ is a Stackelberg equilibrium. Hence, the Pareto optimal insurance contract $(I^*,\pi^*)$ can be obtained from the Stackelberg equilibrium $(I^*,g^*_{I^*})$, where $\pi^* := \Pi_{g^*_{I^*}} \left( I^*(X) \right)$ . \qed

%============================================
%============================================
%============================================
\vspace{0.8cm}
\bibliographystyle{ecta}
\bibliography{ref}
\vspace{0.4cm}
\end{document}